\documentclass[noinfoline,12pt]{imsart}

\RequirePackage[OT1]{fontenc}
\RequirePackage{amsthm,amsmath}
\RequirePackage[numbers]{natbib}
\RequirePackage[colorlinks,citecolor=blue,urlcolor=blue]{hyperref}

\usepackage{verbatim}

\allowdisplaybreaks

\usepackage{fullpage}

\newcommand {\Prob}{\mathbb{P}}

\newcommand{\E}{\mathbb{E}}

\newcommand {\B}{\mathscr{B}}

\newcommand {\cov}{\textrm{Cov}}

\usepackage{subfigure}
\usepackage{layout}

\usepackage{dsfont}
\usepackage[mathscr]{eucal}
\usepackage[toc,page]{appendix}
\usepackage{mathrsfs}
\usepackage{color}
\usepackage{pifont}
\usepackage{bm}
\usepackage{latexsym}
\usepackage{amsfonts}
\usepackage{amssymb}
\usepackage{epsfig}
\usepackage{graphicx}
\usepackage{multirow}

\newtheorem{theorem}{Theorem}
\newtheorem{proposition}{Proposition}

\newtheorem{lemma}{Lemma}
\newtheorem{definition}{Definition}

\newtheorem{assumption}{Assumption}

\newcommand{\transpose}{^{\top}}
\startlocaldefs
\numberwithin{equation}{section}
\theoremstyle{plain}
\endlocaldefs

\begin{document}

\begin{frontmatter}

\title{{\large Recovering Covariance from Functional Fragments}$^*$}

\runtitle{Recovering Covariance from Functional Fragments}

\begin{aug}
\author{\fnms{Marie-H\'el\`ene} \snm{Descary}$^1$\ead[label=e1]{marie-helene.descary@epfl.ch}} \and
\author{\fnms{Victor M.} \snm{Panaretos}$^2$\ead[label=e2]{victor.panaretos@epfl.ch}}

\thankstext{t1}{Research supported by a Swiss National Science Foundation grant.}

\runauthor{M.-H. Descary \& V.M. Panaretos}

\affiliation{University of Geneva and Ecole Polytechnique F\'ed\'erale de Lausanne}

\address{$^1$ Universit\'e du Qu\'ebec \`a Montr\'eal, $^2$ Ecole Polytechnique F\'ed\'erale de Lausanne}


\end{aug}

\begin{abstract} We consider nonparametric estimation of a covariance function on the unit square, given a sample of discretely observed fragments of functional data. When each sample path is only observed on a subinterval of length $\delta<1$, one has no statistical information on the unknown covariance outside a $\delta$-band around the diagonal. The problem seems unidentifiable without parametric assumptions, but we show that nonparametric estimation is feasible under suitable smoothness and rank conditions on the unknown covariance. This remains true even when observation is discrete, and we give precise deterministic conditions on how fine the observation grid needs to be relative to the rank and fragment length for identifiability to hold true. We show that our conditions translate the estimation problem to a low-rank matrix completion problem, construct a nonparametric estimator in this vein, and study its asymptotic properties. We illustrate the numerical performance of our method on real and simulated data. 
\end{abstract}



\end{frontmatter}

\tableofcontents

\newpage

\section{Introduction}\label{sec_intro}

Functional data analysis \citep{hsing-book,ramsay-silver} comprises a broad class of problems and techniques for inferring aspects of the law of a random function $X(t):[0,1]\rightarrow\mathbb{R}$ given multiple realisations thereof. These problems cover the full gamut of statistical tasks, including regression, classification, and testing \citep{wang-review}. A key role is played by the covariance operator of the random function $X(t)$. This is an integral operator with kernel $r(s,t)=\mbox{cov}\{X(s),X(t)\}$, encoding the second-order fluctuations of the random function $X(t)$ around its mean function. This operator, and its associated spectral decomposition, are at the core of many methods obtained via dimension reduction, but also appear in the regularisation of inference problems, which are nearly always ill-posed in the functional case \citep{panaretos-tavakoli-cramer}. 

Consequently, the estimation of the covariance operator associated with $X$ is typically an important first step in any functional data analysis. This is to be done on the basis of independent and identically distributed realisations $X_1,\ldots,X_n$ of the random process $X$. If these are fully observable as continuous curves, then the estimation problem has a simple solution via the empirical covariance $n^{-1}\sum_{i=1}^n\{X_i(s)-\mu_n(s)\}\{X_i(t)-\mu_n(t)\}$, with $\mu_n(s)=n^{-1}\sum_{i=1}^nX_i(s)$. This enjoys several appealing properties, courtesy of the law of large numbers and central limit theorem in Banach space \citep{dauxois}. In practice, however, $X_1,\ldots,X_n$ are not observable as complete trajectories $\{X_i(t):t\in[0,1]\}$. Instead, one only has some finite-dimensional measurements on each function. Typically one observes point evaluations on a grid, and the nature and degree of difficulty of the estimation problem depends upon the structure of the grids $\{t_{ij}\}$, broadly classified as dense or sparse \citep{zhang2016sparse}. 
 
Still, there are cases where even less information is available. In particular, it can happen that each curve is censored, and can only be observed on random subsets $O_i$ of $[0,1]$. These censored curves are referred to as functional fragments. Here too, one can consider regimes coarsely paralleling the dense/sparse observation setting, where one has qualitatively different information on $r(s,t)$ over different subregions of $[0,1]^2$. The first regime, which we call the blanket regime, is such that the typical $O_i$ is a union of sub-intervals with a non-negligeable probability of covering all of $[0,1]$. In this case, information on $r(s,t)$ is available for all $(s,t)\in [0,1]^2$, perhaps of variable amount over different pairs $(s,t)\in[0,1]^2$, depending on the number of $i$ for which $(s,t)\in O_i\times O_i$; in any case the effective sample size available for estimating the covariance is not materially different from $n$ on most of $[0,1]^2$, as illustrated in Fig.~\ref{blanketVSbanded}. The second regime, which we call the banded regime, is such that each $O_i$ is a single interval of length at most $\delta$, for some $\delta>0$ distinctly smaller than 1. Here we have no information on the covariance $r(s,t)$ outside the band $\B_\delta=\{(s,t)\in[0,1]^2:|s-t|\leq\delta\}$. Moreover, the information that we do have on $r(s,t)$ will be reliable only on a strictly narrower band of width $\delta'<\delta$ and centred near $(1/2,1/2)$ since, as illustrated in Fig.~\ref{blanketVSbanded}, the effective sample size is at least halved elsewhere. One thus has to consider the effective $\delta'$ instead of the exact $\delta$, and we will call this the effective bandwidth.

\begin{figure}
\centering
\includegraphics[scale=0.7]{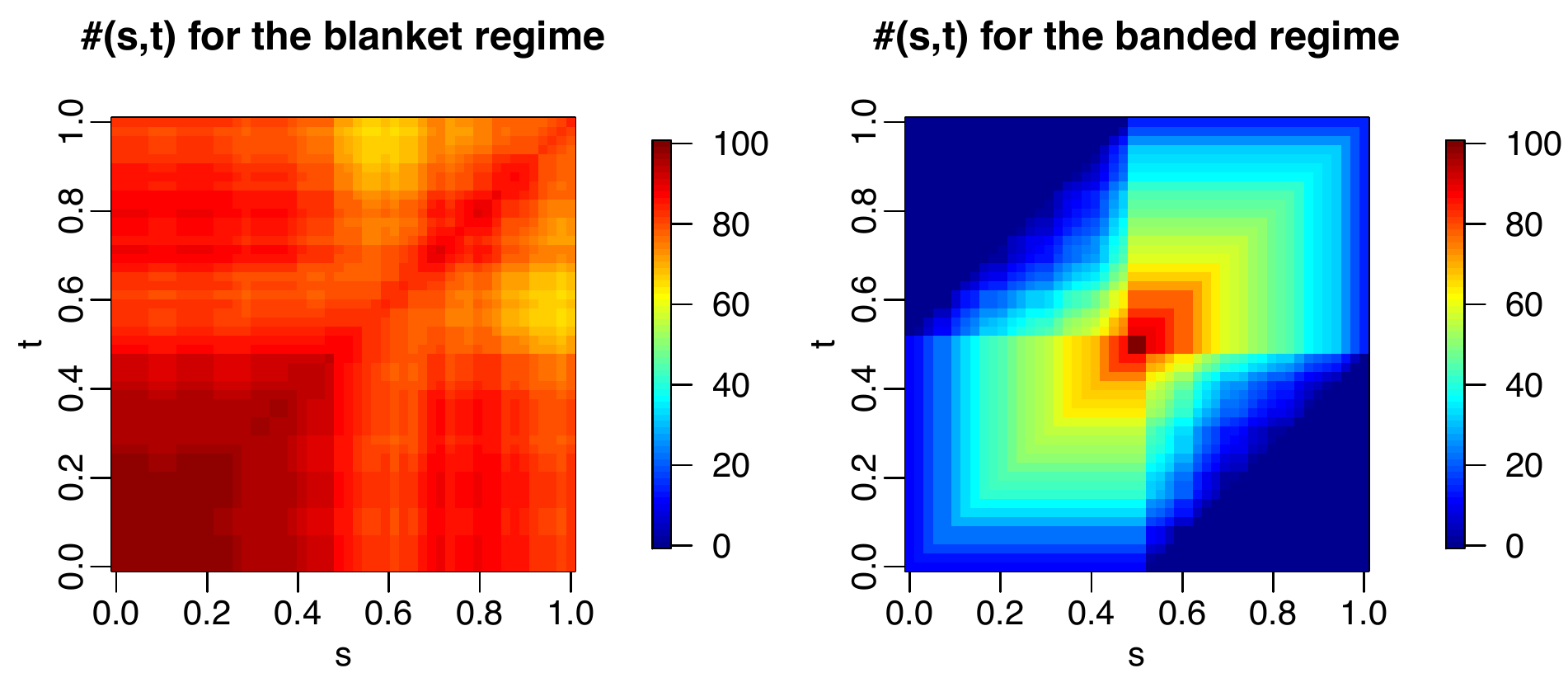}  
\caption{Illustration of the blanket (left) and banded (right) regimes. The entry $(s,t)$ in each plot represents the number of observations out of $n=100$ available to estimate $r(s,t)$, i.e. the value $|i : (s,t)\in O_i\times O_i|$. The $O_i$ were simulated as in \citet[Section 5]{kraus2015components} for the blanket regime and with length of $\delta=0.5$ for the banded regime.}
\label{blanketVSbanded}
\end{figure}

The blanket regime was considered by \cite{kraus2015components}, who provided a nonparametric covariance estimator,  assuming complete observation of the functional data. The band-limited regime is considerably more challenging, and to the best of our knowledge has only been investigated by \citet{delaigle2013classification,delaigle2016approximating}, who discuss the more challenging situation where observation is discrete. In their first paper, they show that despite the low degree of specification, a simple strategy involving gluing neighbouring fragments between curves can offer a way forward, at least when the goal is classification. In their second paper, they model the discrete curve observations as a Markov process. This allows them to complete the discretised curves and discretised covariance by a conditional averaging procedure. The method is appealing and shows good performance in practice, but the model is only posited at the discrete level, and would correspond to a diffusion in the continuum. 
  
Here we investigate whether it is possible to construct a nonparametric estimator of the covariance $r(s,t)$, on the basis of discretely observed fragments in the banded regime. We first study the problem of identifiability, i.e., when can discrete information inside a $\delta$-band uniquely determine the values of the covariance on grid points that are in the censored region $\{(s,t)\in[0,1]^2:|s-t|>\delta\}$. We show that, assuming analyticity of $r(s,t)$, the missing components are uniquely determined, provided the grid size $K$ exceeds a critical threshold that depends on $\delta$ and the rank of $r(s,t)$. Furthermore, we characterise this unique extension as a low-rank matrix completion of the observable banded covariance. This allows us to construct a nonparametric estimator of the covariance function $r(s,t)$ on the entire rectangle $[0,1]^2$, and study its asymptotic properties as dependent on sample size $n$ and grid size $K$. The numerical performance of our approach is investigated in simulated and real functional data. Our method exploits a novel matrix completion framework for functional data analysis recently introduced in \citet{descary2016functional}. While there are strong parallels in our development, the fragmentation setting analysed here is complementary and in some ways more challenging: it deals with nonparametric extrapolation rather than interpolation, based on rather limited data. We will concentrate on the estimation of $r(s,t)$, rather than the prediction of the censored components of the continuous curves $X_1,...,X_n$ from the discrete fragments. Once an estimate of the complete covariance is available, the missing regions can be predicted from the fragments, using best linear prediction, following the techniques of \citet{liebl2013modeling}, \citet{goldberg2014predicting}, and \citet{kraus2015components}.

\section{Problem Statement and Notation}\label{problem_statement}

Consider a continuous random function $X:[0,1]\rightarrow\mathbb{R}$, seen as a random element of the Hilbert space $L^2[0,1]$ comprised of real square-integrable functions, with inner product and norm
$$\langle f,g\rangle_{L^2}=\int_0^1f(t)g(t)dt,\quad \|f\|_{L^2}=\langle f,f\rangle_{L^2}^{1/2}.$$
Assuming that $E(\|X\|_{L^2}^2)<\infty$,  we may define the mean function $E\{X(t)\}=\mu(t)$ and covariance kernel $r(s,t)=\mbox{cov}\{X(s),X(t)\}$. Given a sample $X_1,...,X_n$ of $n$ independent copies of $X$, their natural estimators are the empirical counterparts
$$\mu_n(t)=\frac{1}{n}\sum_{i=1}^{n}X_i(t),\quad  {r}_n(s,t)=\frac{1}{n}\sum_{i=1}^{n}\{X_i(s)-\mu_n(s)\}\{X_i(t)-\mu_n(t)\}.$$
Suppose that each curve $X_i$ is only observed on a random subinterval $O_i \subset [0,1]$ of length $\delta\in(0,1)$. The $\{O_i\}$ are independent and identically distributed, and independent of the $\{X_i\}$. In this case, we can no longer construct the empirical covariance. We can still construct the patched estimator $\tilde{r}_n$ of \citet{kraus2015components}, which keeps track of the amount of information available for each $(s,t)$:
\begin{equation} \label{cov_Kraus}
\tilde{r}_n (s,t) = \frac{I(s,t)}{\sum_{i=1}^nU_i(s,t)}\sum_{i=1}^nU_i(s,t)[\{X_i(s)-\tilde \mu_{n,st}(s)\} \{X_i(t)-\tilde \mu_{n,st}(t) \}],
\end{equation}
where $U_i(s,t)=1(s\in O_i)1(t\in O_i)$ ,  $I(s,t)=1\left\{\sum_{i=1}^nU_i(s,t)>0\right\}$, and
$$ \tilde \mu_{n,st}(t) = \frac{I(s,t)}{\sum_{i=1}^nU_i(s,t)}\sum_{i=1}^nU_i(s,t)X_i(t).$$
\cite{kraus2015components} introduced this estimator in the blanket regime, but in our banded observation regime, (\ref{cov_Kraus}) will no longer be viable, since the length constraint on the $O_i$ implies that we have no data outside the band $\B_{\delta}$. Consequently, $\tilde r_n(s,t)\equiv 0$ on $[0,1]^2\setminus \B_\delta$. An illustration of $r_n$ and $\tilde r_n$ is provided in Fig.~\ref{def_cov}. As pointed out in $\S$\,\ref{sec_intro}, the estimator $\tilde r_n$ is reliable only on a restricted band $\B_{\delta'}$ for $\delta'<\delta$. 

To complicate matters further, the curves will only be measured on a finite grid of points, and this needs to be taken into account, defining discrete $K$-resolution versions of the quantities already introduced. Let $\{t_j\}_{j=1}^{K}$ be a perturbation, potentially random, of a regular grid of $K$ points defined as
\begin{equation*}
(t_1,\ldots,t_K)\in\mathcal{T}_K=\left\{ (x_1,\ldots,x_K)\in \mathbb{R}^K: x_1\in I_{1,K},\dots,x_K\in I_{K,K}\right\},
\end{equation*}
with $\{I_{j,K}\}_{j=1}^{K}$ being the regular partition of $[0,1]$ into intervals of length $1/K$. From $K$ evaluations of a typical curve $X_i$ on this grid, we can define a $K$-resolution representation of $X$, 
$$X^K_i(t)=\sum_{j=1}^K X_i(t_j) 1(t\in  I_{j,K}).$$ 
The covariance of $X^K_i$, which is in fact the $K$-resolution version $r^K$ of $r$, is 
$$ r^K(s,t) = \mbox{cov}\{X^K_i(s),X^K_i(t)\}=\sum_{j,l=1}^K r(t_j,t_l)1\{(s,t)\in I_{j,K}\times I_{l,K}\},$$
and can be summarised by the $K\times K$ matrix coefficient matrix $R^K = \{r(t_j,t_l)\}_{j,l=1}^{K}$. The empirical $K$-resolution covariance kernel $r_n^K(s,t)$ obtained from $K$ discrete measurements on $n$ independent copies of $X $, $\left\{X_{ij}=X_i(t_j): i=1,\ldots,n, j=1,\ldots,K\right\}$, will similarly be defined as the empirical covariance of the $K$-resolution curves $X_1^K,\ldots,X_n^K$,
$$ r_n^K(s,t) = \sum_{j,l=1}^K r_n(t_j,t_l)1\{(s,t)\in I_{j,K}\times I_{l,K}\},$$
summarised by the matrix $R^K_n = \{r_n(t_j,t_l)\}_{j,l=1}^{K}$. This object is inaccessible in the fragmented case. Instead, combining discrete observation and fragmentation, we can only form the discrete analogue of the patched estimator of \cite{kraus2015components}, defined as
\begin{equation}\label{patched_covariance}
 \tilde r_n^K(s,t) = \sum_{j,l=1}^K \tilde{r}_n (t_j,t_l)1\{(s,t)\in I_{j,K}\times I_{l,K}\}, 
 \end{equation}
with matrix representation given by $\tilde R _n^K= \{\tilde r_n(t_j,t_l)\}_{j,l=1}^{K}$. This is precisely the object obtained when replacing $X_i$ by $X_i^K$ in \eqref{cov_Kraus},  and is illustrated in Fig.~\ref{def_cov}.  Since the length of $O_i$ is $\delta$, the set $\{t_j\}_{j=1}^K \cap O_i$ of points on which the curve $X_i$ is observed contains between $\lfloor K\delta \rfloor -1$ and $\lceil K\delta \rceil +1$ points. This implies that the matrix $\tilde R_n^K$ is guaranteed to have non-zero entries only on the band $B_{\delta}=\{ (j,l) : |j-l| < \lfloor K\delta \rfloor -1\}$.
Consequently, it can only be used as an estimator of the banded version of $R^K$, say $P^K_\delta\circ R^K$, where the matrix $P^K_\delta \in \mathbb{R}^{K\times K}$ is defined as $P^K_\delta(j,l) = 1(|j-l| < \lfloor K\delta \rfloor - 1)$ and $\circ$ denotes the element-wise product. Hence, we need to investigate what nonparametric conditions on $r(s,t)$ will suffice for the problem to be identifiable on the basis of the fragmented discrete measurements: when can we uniquely extrapolate $P^K_\delta\circ R^K$ to recover $R^K$?

The assumption that all the curves are observed on the same grid and without any measurement errors is not essential, but allows for a more transparent analysis and presentation below. Extensions to non-common grids, irregular grids, and even measurement errors are treated in Section \ref{sec:irr_grid_meas_error}.

\begin{figure}
\centering
\includegraphics[scale=0.55]{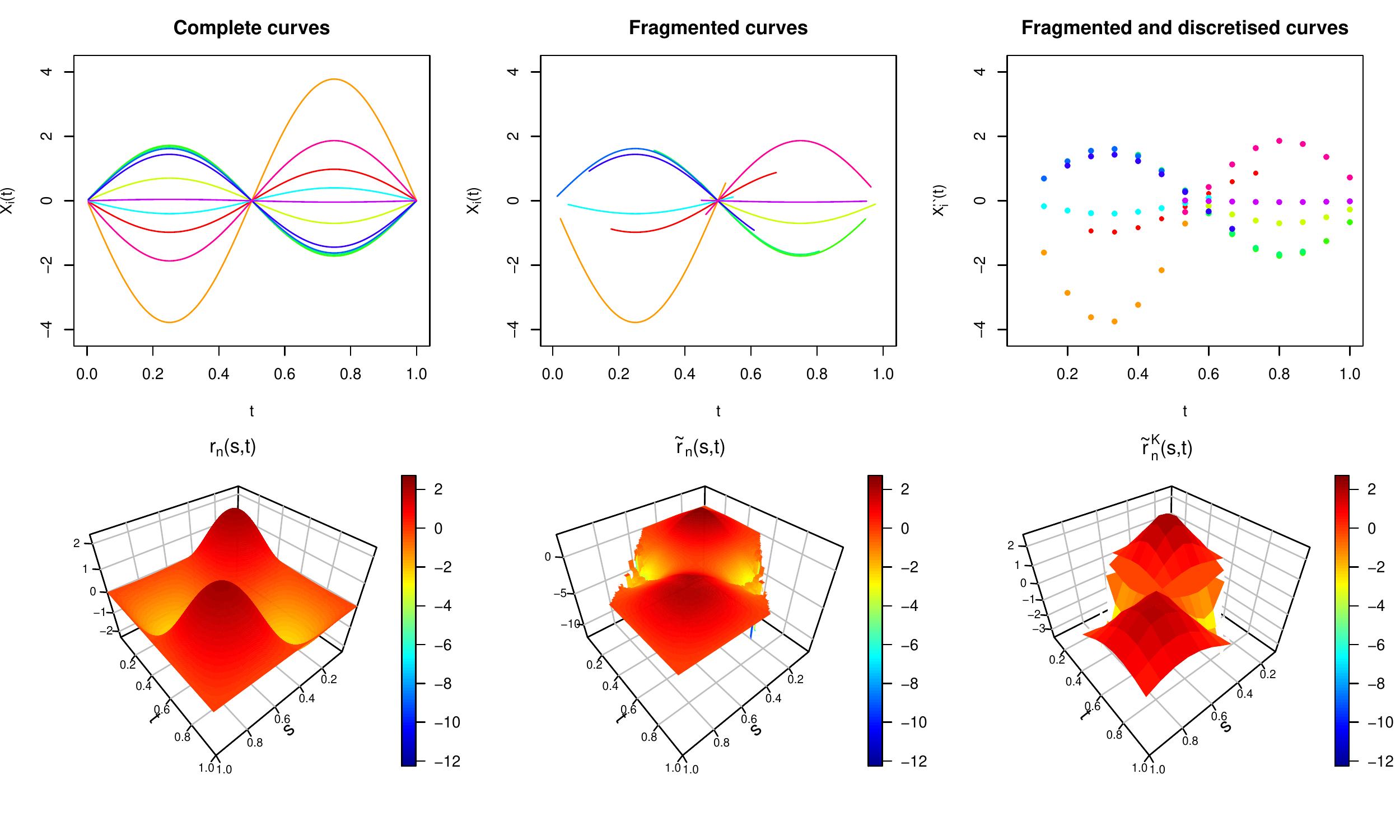} 
\caption{{First row: ten typical curves from a sample of size $n=100$, first under complete continuous observation (left), then under fragmented continuous observation ($\delta=0.5$), and finally under fragmented discrete observation ($\delta=0.5$, $K=15$) on the right. Second row: the associated covariance functions $r_n$, $\tilde{r}_n$, and $\tilde{r}_n^K$.}} \label{def_cov}
\end{figure}

\section{Identifiability} \label{uni_comp}

Despite the weak specification, we wish to impose genuinely nonparametric conditions to ensure identifiability. Still, these will need to be stricter than usual. For even if we were able to perfectly estimate $r(s,t)$ on $\B_{\delta}$ in the continuum, there is a priori no guarantee that the restriction of $r$ on this band extends uniquely to the entire domain of definition. In fact, the problem of extending positive definite functions is notorious in analysis and probability \citep{jorgensen2016extensions}, and connected to the unique extension of characteristic functions via Bochner's theorem. Using this avenue, we can build on classical counterexamples to unique extension of characteristic functions \citep{gnedenko1937fonctions,esseen1945fourier} to show that:

\begin{proposition}\label{prop:counterexample1}
For any $0<\varepsilon<1/2$, there exist $C^{\infty}$ covariance functions $\kappa_1$ and $\kappa_2$ on $(-\pi,\pi)^2$ with common trigonometric eigenfunctions, eigenvalues decaying faster than any polynomial rate, and such that $\kappa_1(x,y)=\kappa_2(x,y)$ if and only if $|x-y|\leq 1-2\varepsilon$. 
\end{proposition}

The proposition illustrates that smoothness alone cannot guarantee unique extension, even if the eigenfunctions are assumed known. A condition that will guarantee unique extension in the continuum is real analyticity \citep{Krantz}. This requires that $r(s,t)$ admit a Fourier series expansion with coefficients that decay not merely faster than any polynomial, but at a rate that is at least geometric. Analyticity will guarantee unique extension from any open band $\B_{\delta}$, by analytic continuation. In light of Proposition \ref{prop:counterexample1}, analyticity is a sharp assumption, despite its strength, if one seeks identifiability with genuinely nonparametric assumptions. Still, in order for this unique extension result to carry over to the discrete case, we will need to impose another condition to relate analyticity of $r(s,t)$ to the matrix properties of $R^K=\{r(t_j,t_l)\}_{j,l=1}^{K}$.  This will be to require that $r(s,t)$ be of finite rank. In summary, we assume:

\begin{assumption}\label{analyticity}
The kernel $r(s,t):[0,1]^2\rightarrow\mathbb{R}$ admits a Mercer decomposition $r(s,t)=\sum_{j=1}^{q}\lambda_j \phi_j(s)\phi_j(t)$ where
the rank $q$ is finite and the orthogonal eigenfunctions $\{\phi_1,\ldots,\phi_q\}$ are real analytic on $(0,1)$.
\end{assumption}

Except for the finiteness requirement, there is no limit on how large $q$ can be, and indeed $q$ is not assumed to be known, nor are the $\varphi_i$ or $\lambda_i$. Since finite-rank analytic covariances are dense among continuous covariances, 
the model can approximate a rich nonparametric class of covariances and is otherwise not restricted. The trade-off for a large value of $q$ manifests itself when we relate $q$ to the identifiability of the $K$-resolution version of the covariance, $r^K$, as summarised in the matrix $R^K$. One expects that the higher the value of $q$, the higher a resolution $K$ would be required. Fortunately, this relationship can be quantified in very precise terms, is linear, and indeed yields rigorous guarantees on identifiability. The proposition below demonstrates that as long as $K>\delta^{-1}(2q+1)$, the matrix $R^K$ is uniquely determined from its entries on the band $B_{\delta}$, by minimal rank completion:

\begin{proposition}\label{unique_completion}
In the notation and framework of Section \ref{problem_statement}, let Assumption \ref{analyticity} hold. 
If $\delta\in(0,1)$ and $K>\delta^{-1}(2q+1)$ then, for almost all grids in $\mathcal{T}_K$, the matrix $R^K = \{r(t_j,t_l)\}_{j,l=1}^{K}$ is the unique solution to the matrix completion problem
\begin{equation} \label{theoretical_min_problem}
\min_{\theta \in \mathbb{R}^{K\times K}}\mathrm{rank}\{\theta\}   \qquad \textrm{subject to} \quad \|P^{K}_\delta\circ (R^K-\theta)\|^2_F=0,
\end{equation}
where $\|\cdot\|_F$ is the Frobenius matrix norm. Equivalently, for almost all grids in $\mathcal{T}_K$, and for all $\tau>0$ sufficiently small,
\begin{equation}\label{theoretical_lagrange}
R^K=\{r(t_j,t_l)\}_{j,l=1}^{K}=\underset{\theta\in \mathbb{R}^{K\times K}}{\arg\min}\left\{K^{-2}\left\| P^K_\delta\circ ( R^K-\theta)\right\|_F^2 +\tau \,\mathrm{rank}(\theta)\right\}.
\end{equation}
\end{proposition}

Even if the rank is finite, analyticity cannot be replaced by smoothness in general. For example, one must rule out locally supported eigenfunctions, or one could construct distinct $C^{\infty}$ covariances of rank three that coincide on $\B_{1/3}$; see Appendix \ref{annex_supp}. 

\section{Estimation of the Covariance Function} \label{estimation}

Proposition \ref{unique_completion} offers a road map for nonparametric estimation of the complete covariance $r(s,t)$ on $[0,1]^2$, via the following three steps:
\begin{enumerate}
\item estimate the banded component $P^K_\delta\circ R^K$ by the empirically constructible matrix $\tilde R _n^K$;
\item solve the problem \eqref{theoretical_lagrange}, with the estimator $\tilde R _n^K$ replacing the estimand $P^K_\delta\circ R^K$;
\item use the optimum obtained, say $\hat R^K_n=\{\hat R^K_n(j,l)\}_{j,l=1}^K$, as the coefficient matrix of a step function $\hat{r}^K_n(s,t)=\sum_{j,l=1}^K \hat R^K_n(j,l)1\{(t,s)\in I_{j,K}\times I_{l,K}\}$, and call this our estimator.
\end{enumerate}
In summary:

\begin{definition}[Covariance Estimator] \label{def_cov_estim}
In the notation of Section \ref{problem_statement}, define an estimator $\hat R_n^K$ of $R^K$ as an approximate minimum of the constrained optimisation problem 

\begin{eqnarray*} 
\min_{ \theta \in \Theta_K} & &  K^{-2} \|  \tilde R^K_n- (P^K_\delta\circ\theta)\|_F^2 +\tau\emph{rank}(\theta) ,  \label{empirical_obj_function}
\end{eqnarray*}
where $\Theta_K$ is the set of $K\times K$ positive semi-definite matrices of trace norm bounded by that of $\tilde R^K_n$, $\tau>0$ is a sufficiently small tuning parameter, and $P^K_\delta(j,l) = 1\{|j-l| < \lfloor K\delta \rfloor - 1\}$. We then define the estimator $\hat{r}_n^K$ of $r$ as the step-function kernel with coefficient matrix $\hat R_n^K$,
\begin{equation*}
\hat{r}^K_n(x,y)=\sum_{j,l=1}^K\hat{R}^K_n(j,l)1\{(x,y)\in I_{j,K}\times I_{l,K}\}.
\end{equation*}
\end{definition}

By approximate minimum we mean that the value of the objective at $\hat R_n^K$ is within $O_{\mathbb{P}}(n^{-1})$ of the value of the overall minimum. Note that the feasible set can be taken to be $\Theta_K$, since no estimator outside that set would be sensible. Despite being a step function, the estimator $\hat{r}^K_n$ is not parametric. If one wishes to have a smooth estimate, it is possible to apply a final post-processing step, and smooth the estimator $\hat{r}^K_n$, with a bandwidth that decays as $K$ increases. Our next theorem considers the performance of the estimator in terms of sample size and resolution.

\begin{theorem}\label{thm:consistency_R}
In the notation and framework of Section \ref{problem_statement}, let Assumption \ref{analyticity} hold. Furthermore, let $E(\|X_i\|^4_{L^2})<\infty$ and suppose that $O_1,\ldots,O_n$ are independent and identically distributed subintervals of $[0,1]$ of length $\delta\in (0,1)$, independent of the $X_1,...,X_n$ and such that $\inf_{|s-t|<\delta}\Prob\{U_1(s,t) =1\}>0$. Define $K^*=\lceil \delta^{-1}(2q+1)\rceil$ to be the critical resolution. Then, for any $K>K^\star$ and for almost all grids in $\mathcal{T}_K$ 
\begin{eqnarray*}
\iint_{[0,1]^2}\left\{\hat{r}^K_n(x,y)-r(x,y)\right\}^2dxdy&\leq&O_{\mathbb{P}}(n^{-1})+{4}K^{-2}\underset{x,y\in[0,1]}{\sup}\|\nabla r(x,y)\|^2_2,
\end{eqnarray*}
for all $\tau>0$ sufficiently small. 
\end{theorem} 

The condition $\inf_{|s-t|<\delta}\Prob\{U_1(s,t) =1\}>0$ asks that the event $\{O_i\times O_i\ni (s,t)\}$ have strictly positive probability as $(s,t)$ ranges over the band $|s-t|<\delta$. Intuitively, this translates to requiring that the $\delta$-fragments be reasonably spread out, for instance, we must avoid the trivial case where $O_i=O$ uniformly in $i$.

\section{Computation}\label{computation}
\subsection{General procedure}
Our estimator can be computed following steps (I)-(IV):
\begin{enumerate}
\item[(I)]  compute the patched covariance matrix $\tilde R _n^K= \{\tilde r_n(t_j,t_l)\}_{j,l=1}^{K}$, as in equation \eqref{patched_covariance};

\item[(II)] solve the optimisation problem
\begin{equation}\label{stepB_min}
\min_{{0 \preceq} \theta \in \mathbb{R}^{K\times K}} K^{-2}\left\|  \tilde R^K_n-(P^{K}_{\delta'}\circ\theta)\right\|^2_F  \quad
\textrm{subject to } \ \mathrm{rank}(\theta)\le i, 
\end{equation} 
where $\delta'<\delta$ is the effective bandwidth, for $i=\{1,\ldots,\lceil K\delta\rceil-3\}$, obtaining minimisers $\hat\theta_1,\ldots,\hat\theta_{\lceil K\delta\rceil-3}$;
\item[(III)]  calculate the fits $\{f(i)=K^{-2}\| \tilde R^K_n-(P^{K}_{\delta'}\circ\hat \theta_i) \|^2_F:i=1,\ldots,\lceil K\delta\rceil-3\}$, and the quantities $f(i)+\tau i$ for some choice of the tuning parameter $\tau>0$; and
\item[(IV)]  determine the $i^*$ that minimises $f(i)+\tau i$, and declare the corresponding optimising matrix $\hat{\theta}_{i^*}$ to be the estimator $\hat R^K_n$.
\end{enumerate}
Step (II), which requires the solution of a rank-constrained least squares problem, and step (III), which requires the selection of a tuning parameter, are discussed in the next two subsections.

\subsection{Rank-Constrained Least Squares}

The fact that the least squares penalty in step (II) involves only some of the matrix entries implies that its best rank-constrained approximation is unavailable in closed form by a simple principal component analysis. Nevertheless we can reformulate the problem as an unconstrained one, by reparametrisation: the space of positive semi-definite matrices of rank $i$ can be spanned by elements of the form $\gamma \gamma\transpose$, where $\gamma\in \mathbb{R}^{i\times K}$. This reduces (\ref{stepB_min}) to
\begin{eqnarray} \label{step1_min}
\min_{\gamma \in \mathbb{R}^{K\times i}} & & K^{-2} \left\|\tilde R^K_n-( P^{K}_{\delta'}\circ \gamma\gamma\transpose)\right\|^2_F.  
\end{eqnarray} 
To solve \eqref{step1_min}, we use the function \texttt{optim} in \texttt{R}, which implements a quasi-Newton method that makes use of the gradient of the objective function. The starting point is the natural candidate $\gamma_0=U_i\Lambda_i^{1/2}$, corresponding to the optimal rank $i$ reduction of $\tilde{R}^K_n$. Specifically, for $U\Lambda U\transpose$ the singular value decomposition of $\tilde R^K_n$, we define $U_i$ as the $n\times i$ matrix obtained by keeping the first $i$ columns of $U$ and $\Lambda_i$ as the $i\times i$ matrix obtained by keeping the first $i$ lines and columns of $\Lambda$. The objective \eqref{step1_min} is convex in $\gamma\gamma\transpose$ but not in $\gamma$ itself, so convergence to a global optimum is not guaranteed. Nevertheless, in our numerical work we observed that the computational implementation was stable, fast and  reliable. Similar stability properties were empirically observed in the band-deleted principal components analysis by matrix completion studied in \citet{descary2016functional}. \citet{wainwright} provide theoretical arguments that gradient-descent type methods can yield good optima with high probability in low-rank matrix completion problems.

\subsection{Scree-Plot Tuning Parameter Selection} \label{subsec:scree-plot}

The role of the tuning parameter $\tau$ is to prevent us from overfitting the banded matrix $\tilde R^K_n$ by selecting too high a rank since $f(i)$ is non-increasing in $i$. The key observation to choosing $\tau$, then, is that selecting a value $\tau$ immediately corresponds to selecting a rank $i_{\tau}$, the rank of the minimum obtained for that $\tau$; in turn, this yields a fit value $f(i_{\tau})$. One can plot $f(i_{\tau})$ as a function of $\tau$, as one would construct a scree-plot in principal components analysis, selecting a $\tau$ by observing an elbow in the plot, or by setting a threshold $\epsilon>0$ and requiring that $f(i_{\tau})<\epsilon$. In fact, since $f(i_{\tau})$ will change only whenever $\tau\mapsto i_{\tau}$ has a jump, we can circumvent $\tau$ entirely, and simply plot the mapping $i\mapsto f(i)$. Under our model assumption, the elbow approach can be theoretically justified: Proposition \ref{unique_completion} shows that if we could use $\tilde R^K$ instead of $\tilde R^K_n$, then we would have $f(i)>0$  when $i\leq q-1$ but $f(q)=0$. One can go beyond the scree-plot in order to determine a rank. For example, one can inspect the resulting optima, and observe how much they differ as the rank increases, or inspect how the eigenvalues of the optima evolve.

\section{Numerical Results}\label{sec:simulations}

To probe the performance of our methodology we perform simulations with two different scenarios for the true covariance function $r(s,t)$. In Scenario A, we set $r(s,t)=\sum_{j=1}^q \lambda_j \phi_j(s) \phi_j(t)$ with the eigenfunctions $\{\phi_j\}$ constructed with constant and sine functions, while in Scenario B, we set $r(s,t)=\sum_{j=1}^q \beta_j \psi_j(s) \psi_j(t)$ with the $\{\psi_j\}$ constructed as Gaussian density functions of mean $\omega_j$ and standard deviation $\sigma_j$. For both scenarios we consider $q=1,2$ and $3$. Table \ref{data_sim} gives more details on the construction of $r$. For a given covariance function $r$, we simulated $100$ samples of $n=200$ centred Gaussian processes $X_i$ such that $\cov\{X_i(s),X_i(t)\}=r(s,t), s,t \in [0,1]$, and we evaluated them on a grid of $K=50$ points; Appendix \ref{add_res} contains additional results. For each sample of curves, we constructed fragmented data by simulating random subintervals $O_1,\ldots,O_n$ of $[0,1]$ of length $\delta$. We consider $\delta=0.5,0.6,0.7,0.8$ and $0.9$. In all simulations, the matrix $P^K_{\delta'}$ is defined with $\delta'=\delta-0.1$ and we implement our method using the true rank, since repeating the optimisation problem through several rank choices and over hundreds of replications would be prohibitive by expensive. Appendix \ref{add_res} contains a study of the performance of our scree-plot approach to select the rank; we found that one typically would select the true rank except in the more challenging cases $\delta\in\{0.5,0.6\}$, where one might select a rank of $2$, though the true rank is $1$ or $3$. For each of the $100$ samples corresponding to a given combination of scenario, of rank and of $\delta$, we compute our estimator $\hat R^K_n$, and then calculate its relative error percentage re$(\hat R^K_n) = (\| \hat R^K_n - R^K\|_F / \| R^K \|_F)\times 100\%$. Table \ref{Normal estimator} gives the median and the first and third quartiles of these $100$ relative errors; we obtain median relative errors of the order of at most $15\%$ once the observation interval length reaches $0.6$, corresponding to an effective length of $0.5$.

\begin{table}[ht] 
\centering
\begin{tabular}{|c|c|c|c|}
\hline
 \multicolumn{2}{|c|}{Scenario A} & \multicolumn{2}{c|}{Scenario B} \\
  \hline
 $\lambda_1 = 1.50$ & $\phi_1(t) = 1 \phantom{\sin(t\pi )}$ & $\beta_1 = 1.50 $ & $\omega_1 =0.5, \sigma_1 = 0.60  $\\
  $\lambda_2 = 0.55$ & $\phi_2(t) = \sin(2\pi t)$ & $\beta_2 = 0.55$ &$\omega_2=0.2, \sigma_2 = 0.25$\\
   $\lambda_3 =0.20 $ & $\phi_3(t) =\sin (4 \pi t) $ & $\beta_3 = 0.20$ & $\omega_3 =0.8, \sigma_3 = 0.20$\\
        \hline
  \end{tabular} \caption{Parameter values for the different simulation scenarios}
  \label{data_sim}
\end{table}

\begin{table}[ht]
\centering
\begin{tabular}{|c|c|c|c|c|}
\hline
Scenario & $\delta$ $(\delta')$& rank $1$ & rank $2$ & rank $3$ \\
  \hline
  \multirow{5}{*}{A} 
  &$0.5$ $(0.4)$ & $16$ $(11, 19)$ & $26$ $(20, 33)$ & $34$ $(31, 38)$ \\ 
  &$0.6$ $(0.5)$& $14$ $(11, 18)$ & $17$ $(13, 21)$ & $17$ $(15, 22)$ \\ 
  &$0.7$ $(0.6)$& $14$ $(10, 18)$ & $15$ $(12, 18)$ & $16$ $(14, 19)$ \\ 
  &$0.8$ $(0.7)$& $12$ $(8, 16)$ &$13$ $(11, 17)$ & $14$ $(12, 18)$ \\ 
  &$0.9$ $(0.8)$& $9$ $(6,13)$ & $12$ $(9,15)$ & $12$ $(10,17)$ \\ 
    \hline
   \multirow{5}{*}{B} 
     &$0.5$ $(0.4)$& $15$ $(11, 18)$ & $17$ $(14, 21)$ & $20$ $(16, 23)$ \\ 
  &$0.6$ $(0.5)$& $13$ $(10, 17)$ & $14$ $(11, 19)$ & $17$ $(14, 22)$ \\ 
  &$0.7$ $(0.6)$& $13$ $(10, 17)$ & $13$ $(10, 17)$ & $15$ $(12, 19)$ \\ 
  &$0.8$ $(0.7)$& $11$ $(8, 14)$ &$12$ $(9, 17)$ & $13$ $(10, 17)$ \\ 
  &$0.9$ $(0.8)$& $9$ $(5,13)$ & $9$ $(7,14)$ & $10$ $(8,13)$ \\ 
      \hline
  \end{tabular}
  \caption{Median of the relative error percentage of our estimators for each scenario, rank and value of $\delta$. The first and third quartiles are in parentheses.}
  \label{Normal estimator}
\end{table}

\section{Data Analysis}

We test our method using daily power spot prices in the German electricity market (www.eex.com). The spot prices are recorded every hour of every working day from January 1st 2006 to September 30th 2008, giving $n=638$ curves observed on a grid of $K=24$ points. The raw dataset and its empirical covariance function $r_n^K$ are depicted in Fig.~\ref{real_data_illu}. Since we are interested in fragmented data, we constructed six functional datasets of size $N=100$ samples of fragmented data from the original dataset, each corresponding to a different value of $\delta$, by simulating random subintervals $O_1,\ldots,O_{638}$ of $[0,1]$ of length $\delta$. To test our method, we estimate the covariance function for each sample, and compare it to the empirical covariance function $r_n$ since we do not know the true covariance function. The empirical covariance function is full rank but its first three eigenfunctions explain more than $90\%$ of the total variance, so we expect our method to work reasonably well even if the condition on $K$ of Proposition \ref{unique_completion} is not exactly satisfied. The first step of our methodology is to pick the rank of our estimator, which is done using the scree plot approach explained in Section \ref{subsec:scree-plot}. The function $f(i)$ ($i=1,\ldots, 8$), obtained for one sample of each set is plotted in Fig.~\ref{real_data_find_r}. After inspection of the figure, we selected $\hat q = 3$ for the six sets of samples, and obtained the associated estimators $\hat R^K_{n,\delta}$. We calculated the empirical relative error percentage $\mathrm{EMPre}(\hat R^K_{n,\delta}) =( \| \hat R^K_n - R_n^K\|_F / \| R_n^K \|_F)\times 100\%$ of each and reported for each set the median of the $100$ resulting empirical relative errors and their first and third quartiles in Table \ref{real_data_rel_err}. Even when we observe only $50\%$ of each curve, we obtained an error smaller than $15\%$. Since the choice of $\hat q$ involves the subjective appreciation of a plot, we also repeated the analysis with $\hat q$ equal to $4$ and $5$. The results are very similar, indicating that mildly overestimating the rank is of little material importance.

\begin{figure}
\centering
\includegraphics[scale=0.8]{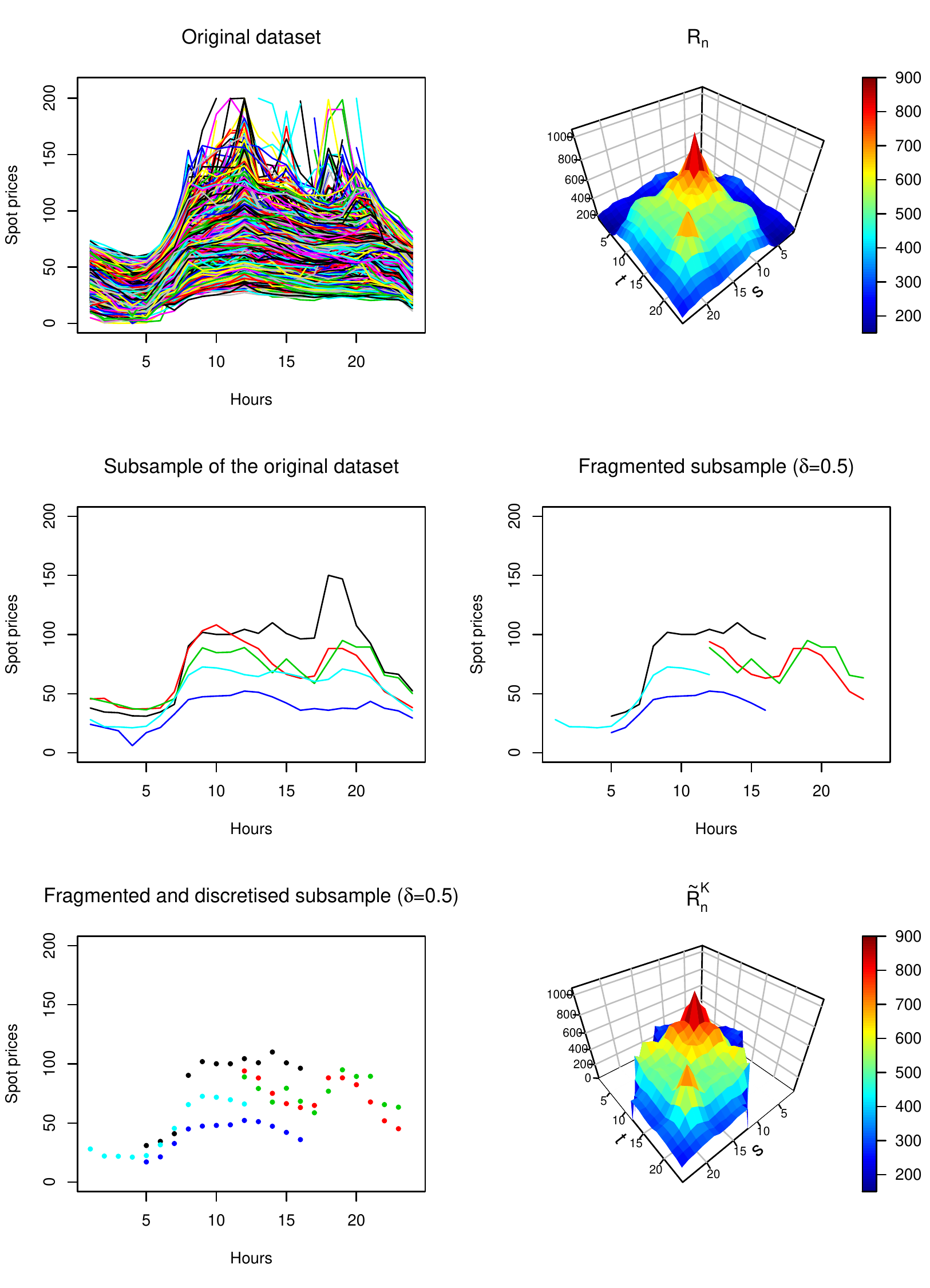} 
\caption{Spot prices in the German electricity market. First row: the complete dataset is depicted on the left, and its empirical covariance function on the right. Second row: a subsample of five complete curves is depicted on the left and a fragmented version with $\delta=0.5$ on the right. Last row: illustration of the discretised version of the fragmented subsample on the left and the empirical covariance matrix $\tilde R_n^K$ for a sample of fragmented curves with $\delta=0.5$ on the right.}
\label{real_data_illu}
\end{figure}

\begin{figure}
\centering
\includegraphics[scale=0.6]{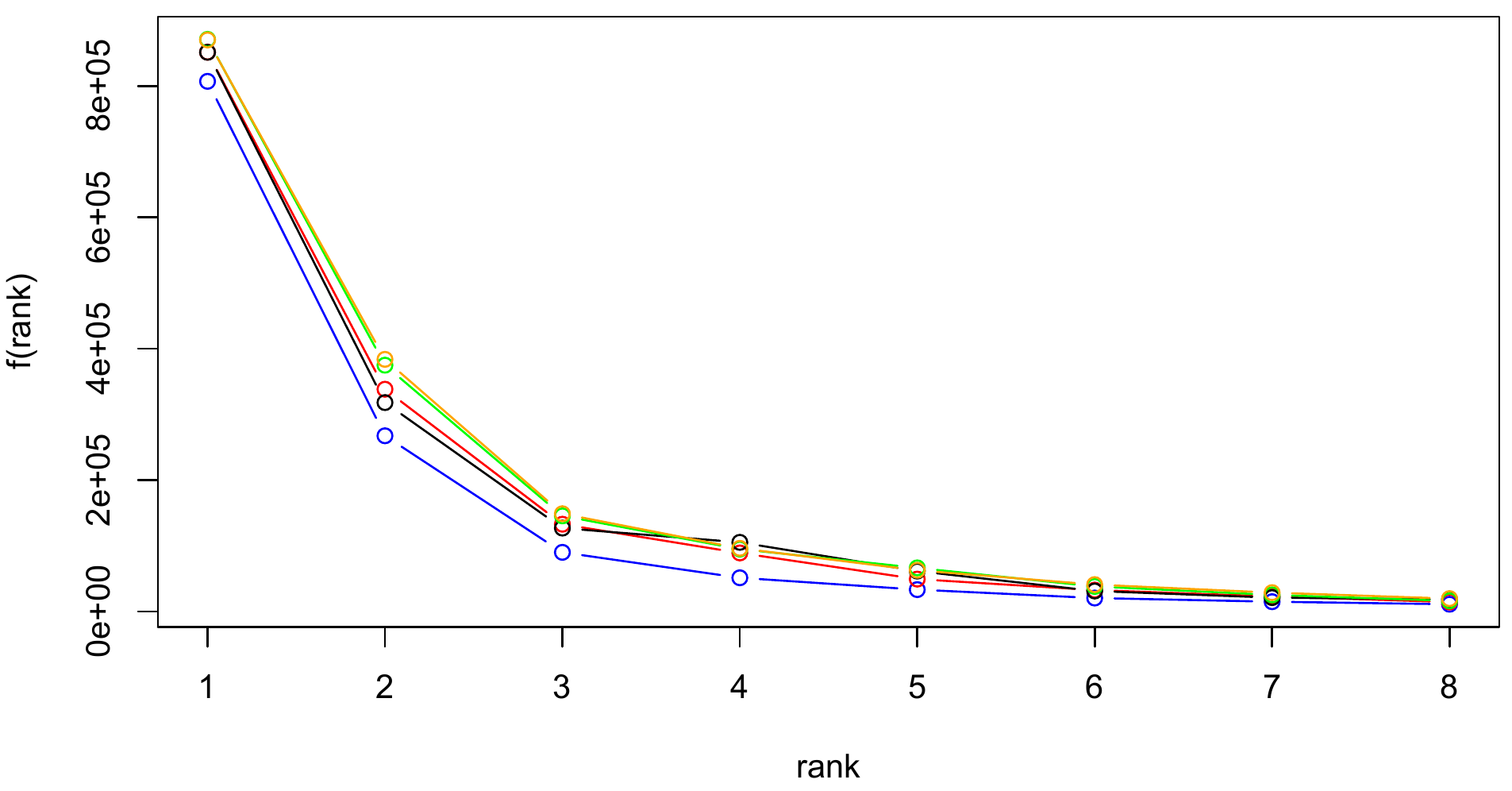}  \\
\caption{Illustration of the scree plot approach to rank selection. The curves in blue, red, black, green and orange correspond to settings with $\delta=0.5, 0.6, 0.7, 0.8$ and $0.9$.}
\label{real_data_find_r}
\end{figure}

\begin{table}[ht]
\centering
\begin{tabular}{|c|c|c|c|}
\hline
$\delta$ $(\delta')$ & $\hat q = 3$ & $\hat q = 4$ & $\hat q = 5$ \\
  \hline
   $0.4$ $(0.3)$& $21$ $(19, 25)$& $24$ $(22, 26)$ & $23$ $(21, 26)$  \\
  $0.5$ $(0.4)$& $14$ $(12, 16)$& $18$ $(16, 19)$ &  $18$ $(16, 19)$  \\ 
  $0.6$ $(0.5)$& $12$  $(11, 13)$ & $12$ $(11, 13)$  & $12$ $(11, 13)$  \\ 
  $0.7$ $(0.6)$& $9$ $(8, 10)$& $8$ $(8, 9)$ & $9$ $(8, 9)$ \\ 
  $0.8$ $(0.7)$& $6$ $(5, 7)$  &$5$ $(5, 6)$& $6$ $(5, 6)$ \\ 
   $0.9$ $(0.8)$& $4$ $(4, 5)$ &$4$ $(3, 4)$& $3$ $(3, 4)$ \\
    \hline
 \end{tabular}
\caption{Median of the empirical relative error percentage of the estimators obtained for different values $\delta$ and $\hat q$. The first and third quartiles are in parentheses.}
\label{real_data_rel_err}
\end{table}

\section{Stability to Departures from Analyticity and Rank Assumptions} \label{sec:non_analytic}

To assess to what extent the performance of our method is stable to perturbations from the analyticity and/or finite rank assumption, we carry out numerical experiments where the true covariance function is neither analytic nor finite rank. We consider centred Gaussian processes with a Mat\'ern covariance function 
$$r_{\textrm{M}}(s,t)=r_{\textrm{M}}(|s-t|)=r_{\textrm{M}}(d)= \sigma^2  \frac{2^{1-\nu}}{\Gamma(\nu)}\left\{(2\nu)^\frac{1}{2}\frac{d}{\rho}\right\}^\nu K_\nu \left\{(2\nu)^\frac{1}{2}\frac{d}{\rho}\right\},\quad s,t\in[0,1],  $$
where $K_\nu$ is the modified Bessel function of the second kind, $\Gamma$ is the gamma function, and the parameters $\nu$ and $\rho$ of the covariance function are non-negative. This is an infinite rank covariance. We consider $\nu=3/2$ and $5/2$ and $\rho=0.5$ and $0.8$. The corresponding Gaussian sample paths will be only $\lfloor \nu\rfloor$ times differentiable \citep[Section 2.7]{Stein_kriging}, so our smoothness settings correspond to sample paths that are at most once, or twice differentiable; this is a severe departure from analyticity, which implies $C^{\infty}$ paths.

Since Mat\'ern covariances are stationary, we also consider a non-stationary scenario. We take the true covariance function $r_{\textrm{M-A}}$ to be equal to $r_{\textrm{M}}+r_{\textrm{A}}$, where $r_{\textrm{A}}$ is defined as in Scenario A with $q=2$. For each combination of scenarios, and values of $\nu$ and $\rho$, we simulated $100$ samples of $n=200$ centred Gaussian processes evaluated on a grid of $K=50$ points and then constructed random fragments of length $\delta=0.5,0.6,0.7,0.8$ or $0.9$. In all simulations, the matrix $P^K_{\delta'}$ is defined with $\delta'=\delta-0.1$ and the rank $\hat q$ is set to $2$, which is conservative/suboptimal in principle. For each sample, we calculate re$(\hat R^K_n)$, the relative error percentage of our estimator, and we report the median and the first and third quartiles of these $100$ relative errors in Table \ref{Res_non_analytic}. 

Despite materially deviating from our assumptions, the method seems quite stable. The relative errors in Table \ref{Res_non_analytic} are comparable to those observed in our earlier simulations. Scenario 2, involving the covariance $r_{\textrm{M-A}}$, lends itself to the most direct comparison, since it can be seen as an additive perturbation of Scenario A with $q=2$ by a non-analytic and infinite-rank covariance. Comparing the results in the last two columns of Table \ref{Res_non_analytic} to those in the first five rows of the column corresponding to $q=2$ of Table \ref{Normal estimator}, we observe only a slight inflation of relative errors. More generally, as $\nu$ increases and the curves become smoother, the performance tends to improve. The typical assumption in functional data analysis is that the observed sample paths are at least of class $C^2$.

\begin{table}[ht]
\centering
\begin{tabular}{|c|c|c|c|c|c|}
\hline
&&\multicolumn{2}{|c|}{Scenario 1 : $r_{\textrm{M}}(s,t)$}&\multicolumn{2}{|c|}{Scenario 2 : $r_{\textrm{M-A}}(s,t)$}\\
  \hline
$\nu $&$\delta$ $(\delta')$& $\rho=0.5$  & $\rho=0.8$ & $\rho=0.5$  & $\rho=0.8$   \\
\hline
\multirow{5}{*}{$3/2$}
&  $0.5$ $(0.4)$ & $33$ (31, 37)& $22$ $(19, 24)$ & $34$ $(31, 38)$ & $30$ $(26, 33)$ \\ 
 &$0.6$ $(0.5)$& $29$ $(25, 32)$ & $21$ $(16, 23)$ & $21$ $(17, 25)$ & $17$ $(14, 20)$  \\ 
 &$0.7$ $(0.6)$& $24$ $(22, 27)$ & $17$ $(13, 21)$ & $16$ $(14, 19)$ & $15$ $(12, 19)$ \\ 
  &$0.8$ $(0.7)$& $19$ $(17, 22)$ &$14$ $(12, 18)$ & $15$ $(12, 17)$ &$14$ $(11, 17)$ \\ 
  &$0.9$ $(0.8)$& $16$ $(13, 19)$ &$11$ $(9, 14)$ & $13$ $(11, 16)$ &$11$ $(9, 15)$ \\
  \hline
\hline
\multirow{5}{*}{$5/2$}
&  $0.5$ $(0.4)$ & $29$ (26, 32)& $20$ $(17, 24)$ & $33$ $(29, 38)$ & $27$ $(24, 31)$ \\ 
 &$0.6$ $(0.5)$& $24$ $(21, 28)$ & $18$ $(15, 22)$ & $19$ $(16, 23)$ & $17$ $(14, 22)$  \\ 
 &$0.7$ $(0.6)$& $22$ $(19, 25)$ & $14$ $(11, 19)$ & $16$ $(13, 20)$ & $15$ $(12, 19)$ \\ 
  &$0.8$ $(0.7)$& $17$ $(14, 21)$ &$13$ $(10, 15)$ & $13$ $(11, 16)$ &$13$ $(10, 16)$ \\ 
  &$0.9$ $(0.8)$& $15$ $(13, 17)$ &$11$ $(8, 15)$ & $12$ $(9, 15)$ &$11$ $(9, 15)$ \\
  \hline
   \end{tabular}
  \caption{Median of the relative error percentage of our estimators for each scenario and values of $\nu$,$\rho$ and $\delta$. The first and third quartiles are in parentheses.}
  \label{Res_non_analytic}
\end{table}

\section{Irregular grids and measurement errors}  \label{sec:irr_grid_meas_error}

Our method can be adapted to cases where each fragment is of differing length, observed on variable and potentially quite irregular grids, and possibly subjected to noise corruption. In the noiseless setting, our observations are discretised fragments of $n$ independent copies of $X$ defined as

\begin{equation}\label{irreg_sample}
\{X_{ij} = X_i(t_{ij}), i=1,\ldots,n, j=1,\ldots,Q_i \},
\end{equation}
with $\{t_{ij}\}_{j=1}^{Q_i}\subset O_i$, where $O_i$ is a random subinterval of $[0,1]$ of random length $\delta_i$, and $Q_i$ is the size of the grid corresponding to the $i$th fragment. For a given positive integer $K$, define as before $\{I_{j,K}\}_{j=1}^K$ to be the regular partition of $[0,1]$ into intervals of length $1/K$ and $R^K$ to be the matrix representation of the $K$-resolution version of the covariance function $r$ of $X$. We define the $K$-resolution patched estimator of $R^K$ based on the sample (\ref{irreg_sample}), as
$$\ddot{R}_n^K(j,l) = \frac{J_{j,l}^K}{\sum_{i=1}^n\sum_{a,b}^{Q_i} W_{j,l}^K(t_{ia},t_{ib})}\sum_{i=1}^n\sum_{a,b}^{Q_i} W_{j,l}^K(t_{ia},t_{ib}) [\{X_{ij}-\ddot{\mu}_n^K(t_j)\}\{X_{il}-\ddot{\mu}_n^K(t_l)\}], $$
where $ W_{j,l}^K(t_{ia},t_{ib}) = 1(t_{ia}\in I_{j,K})1(t_{ib}\in I_{l,K})$, $J_{j,l}^K =1\{\sum_{i=1}^n\sum_{a,b}^{Q_i} W_{j,l}^K(t_{ia},t_{ib}) > 0\}$ and
$$\ddot{\mu}_{n,jl}^K(t_j) = \frac{J_{j,l}^K}{\sum_{i=1}^n\sum_{a,b}^{Q_i} W_{j,l}^K(t_{ia},t_{ib})} \sum_{i=1}^n\sum_{a,b}^{Q_i} W_{j,l}^K(t_{ia},t_{ib}) X_{ij}. $$
The estimator $\hat{r}^K_n$ of $r$ is obtained by replacing $\tilde{R}^K_n$, in steps (I)-(IV) at the beginning of $\S$\,\ref{computation} by $\ddot{R}_n^K$.  In the case of irregular grids, we overlay a regular grid of resolution $K$ on $[0,1]$, and construct the step function estimator of the patched covariance that averages within the $K^2$ induced bins partitioning the domain $[0,1]^2$. The resolution $K$ is commensurate to the $Q_i$, in the sense that if $Q_i\equiv Q$ and $\delta_i\equiv \delta$, then $K=\vartheta\times\lceil Q\delta^{-1}\rceil$, where $\vartheta\in(0,1)$ is chosen such that the average number of points falling in each of the  observed bins is sufficiently large. More generally one calculates empirical averages $\bar{Q}$ and $\bar{\delta}$ and sets $K=\vartheta\times\lceil \bar{Q}{\bar{\delta}}^{-1}\rceil$.  
Since there is now some added flexibility in the choice of $K$, one could even optimise over a choice of $\vartheta$, to exploit any slack in the bias-variance tradeoff induced by the additional averaging within bins, e.g., choosing $\vartheta$ via cross-validation. Implementation is straightforward, and we omit further details.

In the setting where our observations have been corrupted by measurement errors, the sample defined in (\ref{irreg_sample}) becomes $\{X_{ij} = X_i(t_{ij})+\epsilon_{ij}: i=1,\ldots,n, j=1,\ldots,Q_i \}, $
where the $\epsilon_{ij}$ are uncorrelated mean zero random variables of finite variance, and the matrix $\ddot{R}_n^K$ can be defined as before. However, since we know that the diagonal of $\ddot{R}_n^K$ will be corrupted by the noise, we do not use the information that it contains in our matrix completion procedure, we thus replace the matrix $P^K_{\delta'}$ by $\ddot{P}^K_{\delta'}(j,l) = 1(0<|j-l| < \lfloor K\delta' \rfloor - 1)$ in steps (I)-(IV) of Section \ref{computation}.

To probe the performance of our methodology in these new settings, we carry out a simulation study, where we consider two different types of grids. The type 1 grids are mildly irregular; for each grid $\{t_{ij}\}_{j=1}^{Q_i}$ there exists a grid of $K$ points  $\{t_l\}_{l=1}^K\in \mathcal{T}_K$ such that $\{t_{ij}\}_{j=1}^{Q_i}\subset \{t_l\}_{l=1}^K$. Type 2 grids are highly irregular, i.e., each point of the grid $\{t_{ij}\}_{j=1}^{Q_i}$ is uniformly distributed on $O_i$.
In the simulation study, the true covariance function $r(s,t)$ is constructed as in Scenario A (see Section \ref{sec:simulations}), with $q=1,2$ or $3$. We consider $n= 200$ and $400$, and the length $\delta_i$ of the $i$th fragment is uniformly distributed in $(\delta_1,\delta_2)$, with $(\delta_1,\delta_2)=(0.4,0.6),(0.5,0.7),(0.6,0.8)$ or $(0.7,0.9)$. For type 1 grids, we set $Q_i=\lceil K \delta_i\rceil$, where $K=50$, and for type 2 grids, we set $Q_i=\lceil \ddot{K} \delta_i\rceil$, where $\ddot{K}=50$, and $K=4 (5n)^{-1}\sum_{i=1}^n Q_i$. In the setting with measurement errors, the random errors $\epsilon_{ij}$ are simulated as $\textrm{N}(0,1)$. We simulate $100$ samples for each combination of grid type, scenario, sample size $n$, rank $q$ and vector $(\delta_1,\delta_2)$. For each sample we implement our method using the true rank where the matrices $P^K_{\delta'}$ and $\ddot{P}^K_{\delta'}$ are defined with $\delta'=\delta_1$. We then calculate re$(\hat R^K_n)$, the relative error percentage of our estimator, and we report the median and the first and third quartiles of these $100$ relative errors in Table \ref{grid_type1_estimator} and Table \ref{grid_type2_estimator}. 

Comparing the first few lines of Table \ref{grid_type1_estimator} and Table \ref{Normal estimator}, we observe that allowing for variable fragment lengths and variable grids has little impact on estimation when the grids are mildly irregular, and observation is noiseless. Noise has a bigger impact, particularly since it degrades the diagonal of the empirical covariance, where in principle we would have the most information when observing fragments; see Fig.~\ref{blanketVSbanded}. Specifically, we see that to attain similar performance as in the absence of noise, the sample size needs to double. Allowing the observation grid to be highly irregular has a more noticeable impact. Comparing  Tables \ref{grid_type1_estimator} and \ref{grid_type2_estimator}, one sees that to achieve performance comparable to that under mildly irregular grids, with or without noise contamination, one typically needs to double the sample size.

Overall, it would seem that grid/length irregularity and noise contamination do not substantially affect performance, as long as one allows for adapting the sample size to accommodate for the additional layer(s) of ill-posedness. This is no modest feat: we are performing nonparametric estimation of the covariance from functional data that are simultaneously censored, irregularly observed, and noisily observed. With all these layers of ill-posedness combined, it is surprising that nonparametric estimation is feasible at all: for, in the nonparametric case, one ultimately constructs locally parametric estimators, and when we combine these three layers of ill-posedness, there is scant local information to work with.

\begin{table}[ht]
\centering
\begin{tabular}{|c|c|c|c|c|}
 \hline
\multicolumn{5}{|c|}{Type 1 Grid, $n=200$ }\\
 \hline
noise&$(\delta_1, \delta_2)$& rank $1$ & rank $2$ & rank $3$ \\
\hline
\multirow{4}{*}{without}
  &$(0.4, 0.6)$ & $16$ $(14, 21)$ & $26$ $(18, 32)$ & $36$ $(33, 41)$   \\ 
 &$(0.5, 0.7)$& $15$ $(11, 17)$ & $18$ $(15, 23)$ & $19$ $(16, 23)$ \\ 
 &$(0.6, 0.8)$& $13$ $(10, 18)$  & $16$ $(14, 20)$ & $17$ $(15, 21)$ \\ 
  &$(0.7, 0.9)$& $13$ $(10, 17)$ &$14$ $(11, 17)$ & $16$ $(13, 19)$ \\ 
     \hline
   \multirow{4}{*}{with}
    & $(0.4, 0.6)$ & $22$ $(19, 25)$ & $37$ $(41, 45)$ & $45$ $(39, 51)$ \\ 
  &$(0.5, 0.7)$ & $19$ $(17, 22)$ & $25$ $(22, 28)$ & $26$ $(23, 30)$ \\ 
  &$(0.6, 0.8)$ & $18$ $(16, 22)$ & $21$ $(19, 24)$ & $23$ $(20, 26)$ \\ 
  &$(0.7, 0.9)$ & $17$ $(14, 20)$ &$18$ $(16, 21)$ & $22$ $(18, 24)$ \\  
   \hline
   \multicolumn{5}{|c|}{Type 1 Grid, $n=400$ }\\
 \hline
noise&$(\delta_1, \delta_2)$& rank $1$ & rank $2$ & rank $3$ \\
\hline
\multirow{4}{*}{without}
  &$(0.4, 0.6)$ & $11$ $(9, 13)$ & $14$ $(12, 20)$ & $35$ $(31, 38)$ \\ 
 &$(0.5, 0.7)$& $11$ $(9, 14)$ & $12$ $(10, 15)$ & $14$ $(12, 17)$ \\ 
 &$(0.6, 0.8)$& $10$ $(8, 13)$  & $12$ $(10, 14)$ & $12$ $(10, 14)$ \\ 
  &$(0.7, 0.9)$& $9$ $(7, 12)$ &$11$ $(9, 13)$ & $11$ $(9, 13)$ \\ 
     \hline
   \multirow{4}{*}{with}
    & $(0.4, 0.6)$ & $15$ $(13, 18)$ & $25$ $(20, 31)$ & $38$ $(35, 41)$ \\ 
  &$(0.5, 0.7)$ & $15$ $(13, 17)$ & $17$ $(15, 19)$ & $19$ $(17, 21)$ \\ 
  &$(0.6, 0.8)$ & $13$ $(11, 15)$ & $15$ $(14, 17)$ & $16$ $(14, 17)$ \\ 
  &$(0.7, 0.9)$ & $12$ $(10, 14)$ &$14$ $(12, 16)$ & $15$ $(13, 16)$ \\  
   \hline
  \end{tabular}
   \caption{Median of the relative error percentage of our estimators for each scenario, rank, fragments lengths and  type 1 grids. The first and third quartiles are in parentheses.}
  \label{grid_type1_estimator}
\end{table}

\begin{table}[ht]
\centering
\begin{tabular}{|c|c|c|c|c|}
 \hline
\multicolumn{5}{|c|}{Type 2 Grid, $n=200$ }\\
 \hline
noise&$(\delta_1, \delta_2)$& rank $1$ & rank $2$ & rank $3$ \\
\hline
\multirow{4}{*}{without}
  &$(0.4, 0.6)$ & $19$ $(16, 25)$ & $30$ $(25, 34)$ & $33$ $(31, 37)$ \\ 
 &$(0.5, 0.7)$& $19$ $(17, 23)$ & $24$ $(22, 26)$ & $25$ $(23, 29)$ \\ 
 &$(0.6, 0.8)$& $18$ $(15, 21)$  & $23$ $(20, 26)$ & $24$ $(21, 26)$ \\ 
  &$(0.7, 0.9)$& $18$ $(15, 21)$ &$20$ $(18, 23)$ & $22$ $(21, 25)$ \\ 
     \hline
   \multirow{4}{*}{with}
    & $(0.4, 0.6)$ & $23$ $(20, 27)$ & $45$ $(36, 53)$ & $50$ $(43, 63)$ \\ 
  &$(0.5, 0.7)$ & $23$ $(19, 26)$ & $30$ $(26, 33)$ & $33$ $(30, 38)$ \\ 
  &$(0.6, 0.8)$ & $21$ $(18, 24)$ & $26$ $(24, 30)$ & $28$ $(26, 32)$ \\ 
  &$(0.7, 0.9)$ & $21$ $(17, 23)$ &$23$ $(21, 26)$ & $26$ $(24, 29)$ \\  
   \hline
   \multicolumn{5}{|c|}{Type 2 Grid, $n=400$ }\\
 \hline
noise&$(\delta_1, \delta_2)$& rank $1$ & rank $2$ & rank $3$ \\
\hline
\multirow{4}{*}{without}
  &$(0.4, 0.6)$ & $14$ $(12, 16)$ & $21$ $(18, 24)$ & $28$ $(26, 31)$ \\ 
 &$(0.5, 0.7)$& $14$ $(12, 16)$ & $17$ $(15, 20)$ & $19$ $(17, 22)$ \\ 
 &$(0.6, 0.8)$& $14$ $(12, 16)$  & $16$ $(14, 17)$ & $18$ $(16, 20)$ \\ 
  &$(0.7, 0.9)$& $12$ $(11, 15)$ &$15$ $(13, 16)$ & $16$ $(14, 18)$ \\ 
     \hline
   \multirow{4}{*}{with}
    & $(0.4, 0.6)$ & $16$ $(14, 18)$ & $30$ $(24, 38)$ & $40$ $(36, 48)$ \\ 
  &$(0.5, 0.7)$ & $16$ $(13, 18)$ & $21$ $(18, 22)$ & $22$ $(21, 25)$ \\ 
  &$(0.6, 0.8)$ & $16$ $(13, 18)$ & $18$ $(16, 20)$ & $20$ $(19, 23)$ \\ 
  &$(0.7, 0.9)$ & $14$ $(12, 16)$ &$17$ $(15, 19)$ & $18$ $(17, 20)$ \\  
   \hline
  \end{tabular}
   \caption{Median of the relative error percentage of our estimators for each scenario, rank, fragments lengths and  type 2 grids. The first and third quartiles are in parentheses.}
  \label{grid_type2_estimator}
\end{table}

\section{Concluding Remarks}\label{concluding_remarks}

Estimating a covariance from fragments is effectively a nonparametric extrapolation problem: we wish to estimate the long-range covariation properties of a stochastic process by observing only short-range covariation, and without imposing parametric restrictions. Indeed the effective sample size to estimate covariation at distance $u$ decreases rapidly as $u$ increases, and drops to zero when $u$ exceeds $\delta$. Our assumptions, though nonparametric, allow for extrapolation because they impose a sort of rigidity. Assuming analyticity amounts to requiring that all amplitude fluctuations of the stochastic process propagate throughout the global scale $[0,1]$, i.e., there are no purely short-scale variations, which could otherwise be confounded with the long-range variations. This ensures us that the short-range variations that we do observe are genuine sub-samples of the long-range effects that we seek, and can be suitably extrapolated from by analytic continuation. 
The finite-rank restriction complements analyticity by allowing a discrete analytic continuation by matrix completion. 
One can also imagine circumstances where the recovery problem can be relaxed using qualitative knowledge on the boundary behaviour of the stochastic process $X(t)$. For instance, having Dirichlet boundary conditions stipulating that $X(0)$ and $X(1)$ are almost surely fixed but unknown directly translates to the covariance being zero at the boundary of $[0,1]^2$. Such knowledge could in principle be used as an additional constraint in the matrix completion step, though it does not immediately carry over to the matrix factor $\gamma$ in the factorisation $\gamma\gamma\transpose$.

\section*{Acknowledgements}

This research was supported by a grant from the Swiss National Science Foundation.

\appendix\label{supp_frag}

\section{Proofs} \label{sec_proof}

\begin{proof}[Proof of Proposition \ref{prop:counterexample1}]

The proof is based on an explicit construction which smoothes a classical counterexample to unique extension of real characteristic functions from a neighbourhood of zero, due to \cite{esseen1945fourier} and illustrated in Fig.~\ref{counter_ex}. Let $r_1(s,t)=\psi_1(s-t)$ and $r_2(s,t)=\psi_2(s-t)$ be two stationary covariance kernels, where
$$\psi_1(u)=e^{-|u|},\quad \psi_2(u)=\begin{cases}
    \psi_1(u),  & |u|<1, \\
     \psi_1(1)+\psi'_1(1)(|u|-1), &1\leq |u| <1-\psi_1(1)/\psi_1'(1),  \\
      0,& \text{otherwise}.
\end{cases}$$
We note that $r_1$ corresponds to the covariance function of a stationary Ornstein-Uhlenbeck process. \citet[Section 6, p. 22] {esseen1945fourier} shows that $\psi_2$ is a valid characteristic function, and thus by Bochner's theorem it must be positive definite. We now have two covariances such that $r_1=r_2$ on $\mathcal{B}_1$, but clearly $r_1\neq r_2$ outside the band $\mathcal{B}_{1}$. 
\begin{figure}
\centering
\includegraphics[scale=0.3]{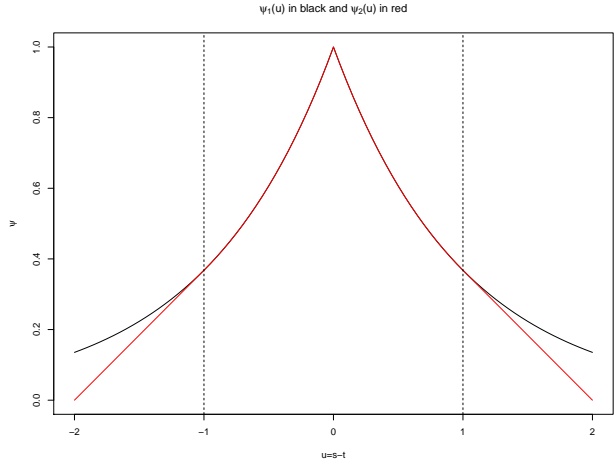} 
\caption{The functions $\psi_1(u)$ (black) and of $\psi_2(u)$ (red).}
\label{counter_ex}
\end{figure}
To construct the covariances $\kappa_1$ and $\kappa_2$, we will smooth Gaussian processes with covariances equal to $r_1$ and $r_2$. In particular, {let $\varepsilon<1/2$}, and consider the bump function 
$$\varphi(u)=1(|u|<\varepsilon)\exp\left\{-\frac{1}{1-\left(u/\varepsilon\right)^2}\right\}.$$
This is $C^{\infty}$ everywhere and obviously supported on $(-\varepsilon,\varepsilon)$. It is the classical example of a non-analytic yet $C^{\infty}$ function. Let $X_1$ and $X_2$ be zero-mean Gaussian processes on $(-\varepsilon-\pi,\pi+\varepsilon)$ with covariances $r_1$ and $r_2$ restricted to $(-\varepsilon-\pi,\pi+\varepsilon)^2$. Define the two processes $Y_1$ and $Y_2$  on $(-\pi,\pi)$ as
$$Y_j(t):=\int_{-\varepsilon}^{\varepsilon}\varphi(u)X_j(t-u)du,\quad j=1,2.$$
\noindent By their very definition, these are $C^{\infty}$ on $(-\pi,\pi)$ with probability 1. Let $\kappa_1$ and $\kappa_2$ be their covariances. We will now show that this specific construction $\kappa_1,\kappa_2$ satisfies all the requirements of Proposition 1. 

We first prove that $\kappa_1=\kappa_2$ on $\mathcal{B}_{1-2\varepsilon}$ but $\kappa_1\neq \kappa_2$, otherwise. Notice that for any $(u,v)$, we may write
\begin{eqnarray*}
\kappa_j(u,v)=E\{{Y}_j(u){Y}_j(v)\}&=&\iint\varphi(s)\varphi(t) E \{X_j(u-s)X_j(v-t)\}dsdt\\
 &=& \iint\varphi(s)\varphi(t)r_j(u-s,v-t)dsdt,
\end{eqnarray*}
for $j=1,2$, and where $s$ and $t$ range in $(-\varepsilon,\varepsilon)$, which implies that $|s-t|<2\varepsilon$. We distinguish two cases:
\begin{enumerate}
\item If $|u-v|\leq1 - 2\varepsilon$, then $|(u-s)-(v-t)|\leq |u-v|+|t-s|<1 - 2\varepsilon +2\varepsilon =1$,
and so $r_1(u-s,v-t)=r_2(u-s,v-t)$. This implies that $\kappa_1(u,v)=\kappa_2(u,v)$, and thus that $\kappa_1$ and $\kappa_2$ are equal on $\mathcal{B}_{1-2\varepsilon}$.
\item If $|u-v| > 1 - 2\varepsilon$, then $\{(s,t) \in (-\varepsilon,\varepsilon)^2:|(u-s)-(v-t)|\geq 1\} \neq \emptyset$. If follows that there exist some $s,t \in (-\varepsilon,\varepsilon)$ for which $r_1(u-s,v-t)$ is not equal to $r_2(u-s,v-t)$, and thus $\kappa_1(u,v)\neq {\kappa}_2(u,v),$ whenever $|u-v|> 1 - 2\varepsilon$.
\end{enumerate}

We now show that $\kappa_1$ and $\kappa_2$ have the same trigonometric eigenfunctions. First note that $\kappa_j$ is stationary: 
\begin{eqnarray*}
\kappa_j(u+a,v+a)&=&\iint\varphi(s)\varphi(t) E \{X_j(u+a-s)X_j(v+a-t)\}dsdt\\
&=&\iint\varphi(s)\varphi(t) E \{X_j(u-s)X_j(v-t)\}dsdt\\
&=&\kappa_j(u,v).
\end{eqnarray*}
By stationarity and symmetry of $\kappa_j$, there exists a symmetric function $\phi_j(u)$ such that $\kappa_j(x,y)=\phi_j(x-y)$. Now let $\sum_{n}\theta^{(j)}_n e^{inu}$ be the Fourier series of $\phi_j(u)$. We have:
$$\kappa_j(s,t)=\phi_j(s-t)=\sum_{n}\theta_n^{(j)} e^{in(s-t)}=\sum_{n}\theta_n^{(j)} e^{ins}e^{-int},$$ 
which shows that the Mercer decomposition of $\kappa_j$ is in terms of complex exponentials, and hence the eigenfunctions of $\kappa_j$ are obtained as products of sines and cosines that do not depend on $j$. 

Finally, we observe that the eigenvalues of $\kappa_j$ are the Fourier coefficients of the function $\phi_j$. Since $Y_j$ is $C^p$ for any $p$ almost surely, so must be $\phi_j(u)=E\{Y_j(0)Y_j(u)\}$, and its Fourier coefficients, i.e., the eigenvalues of ${\kappa}_j$, thus decay at any polynomial rate.
\end{proof}

\begin{proof}[Proof of Proposition \ref{unique_completion}]
We will first prove that $R^K$ is the unique solution of the matrix completion problem \ref{theoretical_min_problem}. Assumption \ref{analyticity} implies that all minors of order $q$ of the matrix $R^K$ are all non-zero, almost everywhere on $\mathcal{T}_K$ \citep[Theorem 4]{descary2016functional}. This implies that the rank of $R^K$ is equal to $q$, and also that the rank of $P^K_\delta\circ R^K$ is at least equal to $q$ since the condition $K>\delta^{-1}(2q+1)$ implies that $P^K_\delta\circ R^K$ and $R^K$ contain a common submatrix of dimension $q\times q$. It now suffices to show that there exists a unique rank $q$ completion of $R^\star$, where $R^\star$ is equal to $R^K$ on the band $B_\delta$ and unknown on $([K]\times[K])/B_\delta$. To do this, we adapt the strategy of \citet[Theorem 2]{descary2016functional}, i.e., we use an iterative procedure to complete the matrix $R^\star$. Due to the band pattern of the observed entries of $R^\star$ and the condition $K>\delta^{-1}(2q+1)$, it is possible to find a submatrix A of $R^\star$ of dimension $(q+1)\times (q+1)$ with only one unknown entry, denoted $x^\star$. The determinant of $A$ can be written as $ax^\star +b$, where $a$ is a minor of order $q$, and thus $a\neq 0$. We look for a rank $q$ completion so det$(A)$ has to be set equal to zero, which leads to the linear equation $ax^\star +b=0$. Since $a\neq 0$, this equation has a unique solution, and it is then possible to impute $x^\star$. The condition $K>\delta^{-1}(2q+1)$ guarantees that it is possible to apply this procedure iteratively until all the missing entries are determined, which allows us to uniquely complete the matrix $R^\star$ such that the resulting matrix be of rank $q$, and thus this unique completion has to be $R^K$ which is of rank $q$. The equivalence between problems \ref{theoretical_min_problem} and \ref{theoretical_lagrange} follows directly from \citet[Proposition 2]{descary2016functional}.\end{proof}

Before proving Theorem 1, we first prove that, when restricted on the band $\B_\delta$, the empirical matrix $\tilde R^K_n$ is a consistent estimator of $R^K$.

\begin{lemma} \label{lemme_kraus}In the notation and framework of Section \ref{problem_statement}, let Assumption \ref{analyticity} hold and suppose that $E(\|X\|^4_{L^2})<\infty$. If there exists $\epsilon>0$ such that $$\sup_{(s,t)\in\B_\delta}\Prob\left\{n^{-1} \sum_{i=1}^nU_i(s,t)\le \epsilon\right\}=O(n^{-2}),$$ 
then $\E\left\| \tilde R^K_n - (P^K_\delta \circ R^K) \right\|^2_F \le K^2 C n^{-1}$, where $C$ is a constant.
\end{lemma}

\begin{proof}
{
Define $Z_i(t) = X_i(t) - E\{X(t)\}$ and let the matrices $\breve{R}^K_n$ and $\bar R^K$ be defined respectively as  
$$\breve{R}^K_n(i,j) =  \breve{r}_n (t_i,t_j)  \textrm{, where } \breve{r}_n (s,t)=  \frac{I(t,s)}{\sum_{i=1}^nU_i(s,t)}\sum_{i=1}^nU_i(s,t) Z_i(s)Z_i(t),$$
and
$$\bar R^K (i,j) = I(t_i,t_j) r(t_i,t_j) . $$
Using these two new matrices, we have the following inequality
\begin{eqnarray}\label{decompo}
E\left\| \tilde R^K_n - (P^K_\delta \circ R^K) \right\|_F^2 &\le& 2 E\left\|\tilde R^K_n - \breve{R}_n^K\right \|_F^2 +4 E\left\| \breve{R}_n^K - \bar R^K\right \|_F^2 \label{decompo} \\
&&+ 4E\left\| \bar R^K -    (P^K_\delta \circ R^K)\right \|_F^2. \nonumber
\end{eqnarray}
Our proof follows essentially the steps of the one of \citet[Proposition 1(b)]{kraus2015components}, so we present it briefly.
Recall that the entries of the matrices $\tilde R^K_n, \breve{R}_n^K, \bar R^K$ and $P^K_\delta$ are defined as zero on $([K]\times[K])/B_{\delta}$. The first term on the right hand side of \eqref{decompo} can be written as  
\begin{eqnarray*}
E\left\|\tilde R^K_n - \breve{R}_n^K \right\|_F^2 &=& E \sum_{(i,j) \in B_\delta} I(t_i,t_j)\left\{\tilde\mu_{t_i,t_j}(t_i) - \mu(t_i)\right \}^2 \left\{\tilde\mu_{t_i,t_j}(t_j) - \mu(t_j)  \right\}^2 \\
&=& E \sum_{(i,j) \in B_\delta}\frac{I(t_i,t_j)}{\left\{\sum_{m=1}^nU_m(t_i,t_j)\right\}^4} \left\{\sum_{m=1}^nU_m(t_i,t_j)Z_m(t_i)\right\}^2 \left\{\sum_{m=1}^nU_m(t_i,t_j)Z_m(t_j)\right\}^2,
\end{eqnarray*}
where the above quantity is dominated by
\begin{equation}
n^{-2} \sum_{(i,j)\in B_\delta} E \left[\frac{n^2I(t_i,t_j)}{\{\sum_{m=1}^nU_m(t_i,t_j)\}^2} \right] \left[\left\{\E Z_1(t_i)^4\E Z_1(t_j)^4\right\}^{1/2} +r(t_i,t_j)^2\right]. \label{bound_C1}
\end{equation}
}

Since $\E(\| X\|^4_{L^2})<\infty$, and since the eigenfunctions of $r$  being real analytic implies that $X$ is analytic, we have that the second term between brackets is bounded. Now note that the first expectation can be written as
\begin{eqnarray*} 
E \left[\frac{n^2I(t_i,t_j)}{\{\sum_{m=1}^nU_m(t_i,t_j)\}^2} \right] & = &\E\left[\frac{n^2I(t_i,t_j)}{\{\sum_{m=1}^nU_m(t_i,t_j)\}^2} 1\left\{n^{-1}\sum_{m=1}^nU_m(t_i,t_j) > \epsilon\right\}\right] \\
&&+ \E \left[\frac{n^2I(t_i,t_j)}{\{\sum_{m=1}^nU_m(t_i,t_j)\}^2} 1\left\{n^{-1}\sum_{m=1}^nU_m(t_i,t_j) \le \epsilon\right\} \right].
\end{eqnarray*}
The first summand is clearly bounded from above by $\epsilon^{-2}$ and the second one by $$n^2\sup_{(s,t)\in\B_\delta}\Prob\left\{n^{-1} \sum_{i=1}^nU_i(s,t)\le \epsilon\right\} = n^2 O(n^{-2}).$$ We can then conclude that $\E\|\tilde R^K_n - \breve{R}_n^K \|_F^2 \le K^2C_1n^{-2}$, where $C_1$ is a constant that bounds from above each term of the summation in (\ref{bound_C1}).

We now prove that $\E\| \breve{R}_n^K - \bar R^K \|_F^2\le K^2 C_2 n^{-1}$. Let $W_i(t,s) = Z_i(t)Z_i(s) - r(t,s)$, then
\begin{eqnarray*}
\E\left\| \breve{R}_n^K - \bar R^K \right\|_F^2 &=&  \E \sum_{(i,j) \in B_\delta} \left\{\frac{I(t_i,t_j)}{\sum_{m=1}^n U_m(t_i,t_j)} \sum_{m=1}^n U_m(t_i,t_j) W_m(t_i,t_j)\right\}^2 \\
& = & \sum_{(i,j) \in B_\delta} \sum_{k=1}^n \sum_{l=1}^n \E \left[\frac{I(t_i,t_j)}{\{\sum_{m=1}^n U_m(t_i,t_j)\}^2}  U_k(t_i,t_j) W_k(t_i,t_j)U_l(t_i,t_j) W_l(t_i,t_j)\right] \\
&=& \sum_{(i,j) \in B_\delta} \frac{1}{n^2}\sum_{k=1}^n \E\left[\frac{n^2I(t_i,t_j)U_k(t_i,t_j)}{\{\sum_{m=1}^n U_m(t_i,t_j)\}^2} \right]\E W_k(t_i,t_j)^2.
\end{eqnarray*}

Using the same arguments as before, we obtain that 
\begin{eqnarray*}
\E\left\| \breve{R}_n^K - \bar{R}^K\right \|_F^2 &\le& n^{-1}  \left[\epsilon^{-2} + n^2\sup_{(s,t)\in\B_\delta}\Prob\left\{n^{-1} \sum_{i=1}^nU_i(s,t)\le \epsilon\right\}  \right] \sum_{(i,j) \in B_\delta} \E\{W_1(t_i,t_j)\} \\
&\le& n^{-1} K^2 C_{2}.
\end{eqnarray*}

Finally, for the third term on the right hand side of (\ref{decompo}) we have
\begin{eqnarray*}
\E\left\| \bar{R}_n^K -  (P^K_\delta \circ R^K) \right \|_F^2 &=& \sum_{(i,j) \in B_\delta}   \E \left[r(t_i,t_j)\{I(t_i,t_j)-1\}\right]^2\\
&=& \sum_{(i,j) \in B_\delta} r(t_i,t_j)^2 \left[1 - \Prob\left\{\sum_{m=1}^n U_m(t_i,t_j)>0\right\}\right]  \\
&=& \sum_{(i,j) \in B_\delta} r(t_i,t_j)^2  \Prob\left\{\sum_{m=1}^n U_m(t_i,t_j)=0\right\}  \\
&\le& \left[  \sup_{(s,t)\in\B_\delta}\Prob\left\{n^{-1} \sum_{i=1}^nU_i(s,t)\le \epsilon\right\} \right]  \sum_{(i,j) \in B_\delta}  r(t_i,t_j)^2\\
&\le& C_3 n^{-2} K^2.
\end{eqnarray*}
Combining the results obtained for the first three term on the right hand side of (\ref{decompo}) leads to the result.

\end{proof}

\begin{proof}[Proof of Theorem \ref{thm:consistency_R}]
We first write 
\begin{eqnarray*}\int_0^1 \! \! \!\int_0^1\{\hat r_n^K(x,y) - r(x,y)\}^2dxdy &\leq& 2\int_0^1 \! \! \! \int_0^1  \{\hat r^K_n(x,y) -r^K(x,y) \}^2dxdy\\
&&+ 2\int_0^1\! \! \! \int_0^1\{r^K(x,y) -r(x,y) \}^2dxdy\\
&=&2K^{-2}\left\| \hat{R}^K_n-R^K \right\|^2_{\mathrm{F}}+2\int_0^1\! \! \! \int_0^1 \{r^K(x,y) -r(x,y) \}^2dxdy.
\end{eqnarray*}
The second term on the right hand side is straightforward to deal with via a Taylor expansion: 
\begin{eqnarray*}
2\int_0^1 \! \! \!\int_0^1 \{r(x,y)-r^K(x,y)\}^2dxdy&=&2\sum_{i,j=1}^{K}\int_{I_{i,K}} \! \!\int_{I_{j,K}}\{r(x,y)-r(t_i,t_j)\}^2dxdy\\
&\leq&2\sum_{i,j=1}^{K}\int_{I_{i,K}}\! \!\int_{I_{j,K}}2K^{-2}\sup_{(x,y)\in I_{i,K}\times I_{j,K} }\left\|\nabla r(x,y)\right\|^2_2\\
&\le&4K^{-2}\sup_{(x,y)\in [0,1]^2}\|\nabla r(x,y)\|^2_2.
\end{eqnarray*}

The main task will now be to show that the term $K^{-2}\left\| \hat{R}^K_n-R^K \right\|^2_{\mathrm{F}}=O_{\mathbb{P}}(n^{-1})$, almost everywhere on $\mathcal{T}_K$. Our strategy will be to follow the approach of \citet[Theorem 3]{descary2016functional} by adapting their argument to the extrapolation rather than the interpolation setting. In particular we will show that $\hat{R}^K_n$ and $\textrm{rank}(\hat{R}^K_n)$ are consistent estimators of $R^K$ and $q$ respectively, and we will then use \citet[Theorem 3.4.1]{emp_pro}.

We assume without loss of generality that the data have been rescaled so that $K^{-1}\mathrm{trace}(\tilde R^K_n)=1$.  Define $\Theta_K$ the space of  $K\times K$ positive semi-definite matrices of trace at most $K$ and consider the functionals

$$\mathbb{S}_{n,K}:\Theta_{K}\rightarrow [0,\infty),\qquad\mathbb{S}_{n,K}(\theta)=\underset{\mathbb{M}_{n,K}(\theta)}{\underbrace{K^{-2}\|{P}^{K}_\delta \circ(\theta-\tilde R_{n}^{K})\|^2_{\mathrm{F}}}}+\tau \mathrm{rank}(\theta),$$
$$S_{K}:\Theta_{K}\rightarrow [0,\infty),\qquad S_{K}(\theta)=\underset{M_K(\theta)}{\underbrace{K^{-2}\|P^K_\delta\circ(\theta-R^{K})\|^2_{\mathrm{F}}}}+\tau \mathrm{rank}(\theta).
$$
Note that, since $K>\delta^{-1}(2q+1)$, Theorem \ref{unique_completion} implies that for almost all grids, $R^K$ is the unique minimizer of $S_{K}$, for all $\tau>0$ sufficiently small. From now on, fix such a grid, and let $\tau>0$ be sufficiently small. 

To establish that $\hat R^K_n$ is consistent for $R^K$, we observe that

\begin{eqnarray*}
 \left|\mathbb{S}_{n,K}(\theta)-S_{K}(\theta)\right| &=& \left| \mathbb{M}_{n,K}(\theta)-M_{K}(\theta)\right| \\
&=& K^{-2} |~\|P^{K}_\delta \circ (\theta - \tilde R_{n}^{K})\|_{F}^{2} - \|P^{K}_\delta \circ (\theta - R^{K})\|_{F}^{2}| \\
&\leq& K^{-2} |~\|P^{K}_\delta \circ (\theta - \tilde R_{n}^{K})\|_{F} - \|P^{K}_\delta \circ (\theta - R^{K})\|_{F}|\\
&&\qquad\quad\times(\|P^{K}_\delta \circ (\theta - \tilde R_{n}^{K})\|_{F} + \|P^{K}_\delta \circ (\theta - R^{K})\|_{F}) \\
&\leq& K^{-2} \|P^{K}_\delta \circ (\tilde R_{n}^{K} - R^{K})\|_F~(2\|\theta\|_{F} + \|\tilde R^{K}_n\|_F + \|R^{K}\|_F).
\end{eqnarray*}

Consequently $\sup_{\theta\in\Theta_K} \left|\mathbb{S}_{n,K}(\theta)-S_{K}(\theta)\right|\stackrel{n\rightarrow\infty}{\rightarrow}0$ almost surely, and by lower semicontinuity of $S_{K}(\theta)$ and the fact that it is uniquely minimized at at $R^K$, we obtain consistency of $\hat{R}^K_n$ for $ R^K$, noting that $\hat R^K_n\in\Theta_K$ \citep[Corollary 3.2.3]{emp_pro}. 

We will next argue by contradiction in order to establish that $\mathrm{rank}(\hat{{R}}^K_n)$ is consistent for the true rank. Assuming that this is not the case, there exist $\epsilon > 0$, $\delta > 0$ and a subsequence $\{n_{j}\}$ such that $\Prob\{|\mathrm{rank}(\hat{R}^K_{n_{j}}) - q| > \epsilon\} > \delta$ for all $j \geq 1$. Therefore, $\Prob\{\mathrm{rank}(\hat{R}^K_{n_{j}}) \neq q\} > \delta$ for all $j \geq 1$, and there exist possibly two subsequences $\{j_{l}\}$ and $\{k_{l}\}$ such that $\Prob\{\mathrm{rank}(\hat{R}^K_{j_{l}}) > q\} > \delta/2$ and $\Prob\{\mathrm{rank}(\hat{R}^K_{k_{l}}) < q\} > \delta/2$ for all $l \geq 1$. The second case is impossible by consistency of $\hat{R}^K_n$ and the fact that matrices of rank at most $q-1$ comprise a closed set. For the first case, since $\hat{R}^K_{j_{l}}$ converges to $R^K$ in probability, we can extract a further subsequence $\{j_{l_{m}}\}$ such that $\mathrm{rank}(\hat{R}^K_{j_{l_{m}}}) > q$ for all $m \geq 1$ and $\hat{R}^K_{j_{l_{m}}}$ converges to $R^K$ as $m \rightarrow \infty$. It can be assumed that $\Prob\{\mathrm{rank}(\hat{R}^K_{j_{l_{m}}}) > q\} > \delta/2$ for all $m \geq 1$, and $\hat{R}^K_{j_{l_{m}}}$ converges to $R^K$ as $m \rightarrow \infty$ almost surely (or take further subsequences). Hence, the probability of the set where both of these events hold has is at least $\delta/2$. If we restrict to this set, and use the fact that $\hat{R}^K_{j_{l_{m}}}$ is a minimizer, we obtain
\begin{eqnarray} \label{s1}
\mathbb{M}_{n,K}(\hat{R}^K_{j_{l_{m}}}) + \tau(q+1) &=& K^{-2} \|P^K_\delta \circ (\hat{R}^K_{j_{l_{m}}} - \tilde R_{n}^K)\|_{F}^{2} + \tau(q+1) \nonumber \\
&\leq& K^{-2} \|P^K_\delta \circ (\hat{R}^K_{j_{l_{m}}} - \tilde R_{n}^K)\|_{F}^{2} + \tau\mathrm{rank}(\hat{R}^K_{j_{l_{m}}}) \nonumber \\
&\leq& \inf_{\theta \in \Theta_{K} \ : \ \mathrm{rank}(\theta) = q} \{K^{-2} \|P^K_\delta \circ (\theta - \tilde R_{n}^K)\|_{F}^{2} + \tau\mathrm{rank}(\theta)\} \nonumber \\
&=& \inf_{\theta \in \Theta_{K} \ : \ \mathrm{rank}(\theta) = q} K^{-2} \|P^K_\delta \circ (\theta - \tilde R_{n}^K)\|_{F}^{2} + {\tau}q \nonumber \\
&=& \inf_{\theta \in \Theta_{K} \ : \ \mathrm{rank}(\theta) = q} \mathbb{M}_{n,K}(\theta) + {\tau}q,
\end{eqnarray}
for all $m \geq 1$. But $\sup_{\theta \in \Theta_{K}} |\mathbb{M}_{n,K}(\theta) - M_{K}(\theta)| \rightarrow 0$ almost surely, so $\mathbb{M}_{n,K}(\hat{R}^K_{j_{l_{m}}}) - M_{K}(\hat{R}^K_{j_{l_{m}}}) \rightarrow 0$, while continuity also yields $M_{K}(\hat{R}^K_{j_{l_{m}}}) \rightarrow M_{K}(R^K) = 0$. It must therefore be that $\mathbb{M}_{n,K}(\hat{R}^K_{j_{l_{m}}}) \rightarrow 0$. On the other hand, the sequence of functions $\mathbb{M}_{n,K}(\theta)$ are almost surely  equi-Lipschitz continuous on the set $\{\theta \in \Theta_{K} \ : \ \mathrm{rank}(\theta) = q\}$, so that by uniform convergence
$$ \inf_{\theta \in \Theta_{K} \ : \ \mathrm{rank}(\theta) = q} \mathbb{M}_{n,K}(\theta) \rightarrow \inf_{\theta \in \Theta_{K} \ : \ \mathrm{rank}(\theta) = q} M_{K}(\theta) = 0.$$
Combining these and and equation \eqref{s1} leads to the contradiction that $\tau \leq 0$. It follows that $d(\hat{R}^K_{n},R^K) \rightarrow 0$ in probability as $n \rightarrow \infty$, where
$$d^2(\theta,R^K)={K^{-2}\|\theta-R^K\|^2_F + {\tau}|\mathrm{rank}(\theta)-\mathrm{rank}(R^K)|}.$$
In order to establish the rate, define
$$
\Delta(\theta)=S_{K}(\theta)-S_{K}(R^K)= K^{-2}\| P^K_\delta \circ(\theta-R^K)\|^2_{\mathrm{F}} +\tau (\mathrm{rank}(\theta)-q).
$$
Let $\eta^{2} < \tau$ and observe that, for any $\theta$ with $\mathrm{rank}(\theta) \neq q$, it must be that $d^{2}(\theta,R^K) \geq \tau|\mathrm{rank}(\theta) - q| \geq \tau > \eta^{2}$, which yields  $d(\theta,R^K) > \eta$. It follows that no matrix $\theta$ with $\mathrm{rank}(\theta) \neq q$ can satisfy $\gamma/2 < d(\theta,R^K) < \gamma$ for $\gamma < \eta$. Hence,
\begin{eqnarray*}
\inf_{\theta \in \Theta_{K} \ : \ \gamma/2 < d(\theta,R^K) < \gamma} \Delta(\theta) = \inf_{\theta \in \Theta_{K} \ : \ \gamma/2 < d(\theta,R^K) < \gamma, \ \mathrm{rank}(\theta) = q} \Delta(\theta).
\end{eqnarray*}

We will show that the latter quantity is bounded below by $\alpha_{0}\gamma^{2}$, where $\alpha_{0}>0$ and $\gamma < \eta$, for $\eta>0$ sufficiently small, or equivalently that
\begin{eqnarray} \label{e3}
\inf_{\theta \in \Theta_{K} \ : \ \gamma^{2}/4 < \|\theta - R^K\|_{F}^{2} < \gamma^{2}, \ \mathrm{rank}(\theta) = q} \|P^K_\delta \circ (\theta - R^K)\|_{F}^{2} > \alpha_1\gamma^{2},
\end{eqnarray}
for some $\alpha_{1} > 0$. This we do again by contradiction. Fix any $\theta$ with $\mathrm{rank}(\theta) = q$ and $\|\theta - R^K\|^{2}_{F} > d$, setting $d=\gamma^2/4$ to simplify things. Assume that $\|P^{K}_\delta \circ (\theta - R^K)\|^{2}_{F} < {\beta}d$, for some $\beta \in (0,1/2)$ . Writing $\theta = R^K + A + B$, where $A = P^{K}_\delta \circ A$ and $P^{K}_\delta \circ B = 0$ (by defining $A = P^{K}_\delta \circ (\theta - R^K)$ and $B = \theta - R^K - A$), we have that whenever $\|P^{K}_\delta \circ (\theta - R^K)\|^{2}_{F} < {\beta}d$ for some $\beta\in(0,1/2)$, it must be that $\|A\|^{2}_{F} < {\beta}d$ and $\|A + B\|^{2}_{F} = \|A\|^{2}_{F} + \|B\|^{2}_{F} > d$. Consequently $\|B\|^{2}_{F} > (1-\beta)d > d/2$ so there must exist an element $(j,k) \notin B_\delta$ such that $|B_{j,k}| > \sqrt{d/(2c_{K})}$, where $c_{K}$ is the total number of elements in $([K]\times [K])\setminus B_\delta$. Observe that $\theta_{j,k} = R^K_{j,k} + B_{j,k}$.

Since all possible minors of $R^K$ of order $q$ are non-zero, the same will be true in an $\eta$-neighbourhood of $R^K$, which includes $\theta$,  for sufficiently small $\eta$. Denote the rows and columns of such an $q\times q$ sub-matrix of $R^K$, say $C_K$, by $\{p_{1},p_{2},\ldots,p_{q}\}$ and $\{r_{1},r_{2},\ldots,r_{q}\}$, respectively. Exploiting the structure of the band, choose this sub-matrix in such a way that the sub-matrix elements and the entries $\{(j,r_l) : 1 \leq l \leq q\}$ and $\{(p_l,k) : 1 \leq l \leq q\}$ lie inside $B_\delta$. Consider the sub-matrix of order $q$ of $\theta$, say $E$, by taking the same rows and columns as in $C_K$. Define the sub-matrix $F$ (resp. D) of order $(q+1)$ obtained by adjoining to $E$ (resp. to $C_K$), the elements ${q}_{1}=(\theta_{j,r_1},\ldots,\theta_{j,r_q})'$, ${q}_{2}=(\theta_{p_1,k},\ldots,\theta_{p_q,k})'$ and $\theta_{j,k}$ (resp. the elements ${c}_{1}=(R^K_{j,r_1},\ldots,R^K_{j,r_q})'$, ${c}_{2}=(R^K_{p_1,k},\ldots,R^K_{p_q,k})'$ and $R^K_{j,k}$). So, 
\begin{eqnarray*}
F = 
\begin{bmatrix}
\theta_{j,k} & {q}_{1}' \\
{q}_{2} & E
\end{bmatrix} \quad \textrm{and } D=
\begin{bmatrix}
R^K_{j,k} & {c}_{1}' \\
{c}_{2} & C_K
\end{bmatrix}.
\end{eqnarray*} 
Then, for $\eta$ sufficiently small, we have that
$$|B_{j,k}| = |{q}_{1}'E^{-1}{q}_{2} - {c}_{1}'C_{K}^{-1}{c}_{2}|<\kappa \|P_\delta^K\circ(\theta-R^K)\|_F<\kappa\sqrt{\beta d},$$
by the fact that the map $({{q}}_1,{{q}}_2,E)\mapsto {{q}}_{1}'E^{-1}{{q}}_{2}$ is locally Lipschitz at any $({{c}}_1,{{c}}_2,C_K)$ as constructed above. For $\beta$ chosen to be sufficiently small, this contradicts the fact that $|B_{j,k}| > \sqrt{d/(2c_{K})}$. We conclude that, for some $\beta\in (0,1/2)$ sufficiently small, it holds that $\|P^{K}_\delta \circ (\theta - R^K)\|^{2}_{F} > {\beta}d$ if $\theta$ is a rank $r$ matrix with $\|\theta - R^K\|^{2}_{F} > d$.  

Next, define
\begin{eqnarray*}
D(\theta)&=&\mathbb{S}_{n,K}(\theta)-S_{K}(\theta)-\mathbb{S}_{n,K}(R^K)+S_{K}(R^K)\\
&=&\mathbb{M}_{n,K}(\theta)-M_{K}(\theta)-\mathbb{M}_{n,K}(R^K)+M_{K}(R^K).
\end{eqnarray*}
Expanding $(\mathbb{M}_{n,K} - M_{K})$ in a first order Taylor expansion with Lagrange remainder, around $R^K$, gives 
\begin{eqnarray*}
D(\theta)&=&\langle\mathbb{M}'_{n,K}(\tilde\theta),\theta-R^{K}\rangle_{\mathrm{F}}-\langle M'_{K}(\tilde\theta),\theta-R^{K}\rangle_{\mathrm{F}}\\
&=&K^{-2} \langle 2P_\delta^{K}\circ(\tilde\theta-\tilde R_{n}^K),\theta-R^K\rangle_{\mathrm{F}}-K^{-2} \langle 2P_\delta^{K}\circ(\tilde\theta-R^K),(\theta-R^K)\rangle_{\mathrm{F}}\\
&=&K^{-2} \langle  2P_\delta^{K}\circ\tilde\theta-2P_\delta^{K}\circ\tilde\theta-2P_\delta^{K}\circ \tilde R_{n}^{K}+2P_\delta^{K}\circ R^{K},\theta-R^{K}\rangle_{\mathrm{F}}\\
&\leq& 2K^{-2} \left \|P_\delta^{K}\circ(\tilde R_{n}^{K}-R^{K})\right \|_{\mathrm{F}}\left \|\theta-R^{K}\right\|_{\mathrm{F}},
\end{eqnarray*}
for a certain $\tilde p \in [0,1]$ and $\tilde\theta =\tilde p{R}^K+(1-\tilde p)\theta$. Note that the fact that the subintervals $\{O_1,\ldots,O_n\}$ are independent and identically distributed and such that $\inf_{|s-t|<\delta}\Prob\{U_1(s,t) =1\}>0$ holds implies that there exists $\epsilon>0$ such that 
$$\sup_{(s,t)\in\B_\delta}\Prob\left\{n^{-1} \sum_{i=1}^nU_i(s,t)\le \epsilon\right\}=O(n^{-2}),$$
and thus by Lemma \ref{lemme_kraus}, we know that $ \E\left \|P^{K}\circ(\tilde R_{n}^{K}-R^{K})\right \|^2_{\mathrm{F}} \le  C K^2 n^{-1}$.

Consequently, by the choice of $\eta$ in relation to $\tau$, it follows that
\begin{eqnarray} \label{e31}
\E\left\{\sup_{\theta \in \Theta_{K} : d(\theta,R^K) < \gamma} |D(\theta)|\right\} &=& \E\left\{\sup_{\theta \in \Theta_{K} : d(\theta,R^K) < \gamma, \mathrm{rank}(\theta) = q} |D(\theta)|\right\} \nonumber \\
&=& \E\left\{\sup_{\theta \in \Theta_{K} : K^{-1}\|\theta - R^K\|_{F} < \gamma} |D(\theta)|\right\} \nonumber \\
&\leq& 2\gamma K^{-1} \E\left \| P^K_\delta \circ (\tilde R_{n}^{K}-R^{K})\right \|_{\mathrm{F}}  \ \leq \ 2\gamma \left(\frac{{C}}{{n}}\right )^{1/2}.
\end{eqnarray}
Combining \eqref{e3} and \eqref{e31}, we conclude that $nK^{-2}\|\hat{R}^K_n - R^K\|_{F}^{2} = O_{\mathbb{P}}(1)$ since $\hat{R}^K_n$ is an approximate minimizer of $\mathbb{S}_{n,K}$  \citep[Theorem 3.4.1]{emp_pro}.
\end{proof}

\section{An additional counterexample} \label{annex_supp}

It was remarked at the end of Section \ref{uni_comp} that one can construct distinct rank three covariances that are $C^{\infty}$ and nevertheless agree on $\mathcal{B}_{1/3}$. We now construct such an example. Once again, consider the bump function 
$$\varphi(u)=1(|u|<1/J)\exp\left\{-\frac{1}{1-\left(Ju\right)^2}\right\},\quad u\in \mathbb{R},$$
which is $C^{\infty}$ everywhere and supported on $(-1/J,1/J)$. Take $J=6$, and define
$$X(t)=\xi_1 \underset{\varphi_1(t)}{\underbrace{\varphi(t-1/6)}}+\xi_2 \underset{\varphi_2(t)}{\underbrace{\varphi(t-1/2)}}+\xi_3 \underset{\varphi_3(t)}{\underbrace{\varphi(t-5/6)}},\quad t\in [0,1],$$
and 
$$Y(t)=\xi_1 \varphi_1(t)+\xi_2 \varphi_2(t)+\{\lambda^{1/2}\xi_1+(1-\lambda)^{1/2}\xi_3\} \varphi_3(t),\quad t\in [0,1], \lambda\in(0,1),$$
where the $\xi_i$ are independent and distributed as $N(0,1)$. Let $\kappa_1$ and $\kappa_2$ be the covariance functions of $X$ and $Y$ respectively, and observe that
$$ \kappa_1(s,t)=\sum_{i=1}^3 \phi_i(s)\phi_i(t),$$
and
\begin{eqnarray*}
 \kappa_2(s,t)&=&\sum_{i=1}^3 \phi_i(s)\phi_i(t) + \lambda^{1/2}\{\phi_1(t)\phi_3(s) + \phi_1(s)\phi_3(t) \}\\
 &=& \kappa_1(s,t) + \lambda^{1/2}\{\phi_1(t)\phi_3(s) + \phi_1(s)\phi_3(t) \}.
 \end{eqnarray*}
Since the functions $\phi_i$ have disjoint supports of length $1/3$, we have that $\lambda^{1/2}\{\phi_1(t)\phi_3(s) + \phi_1(s)\phi_3(t) \}$ is non-zero only if $(s,t)\in (2/3,1)\times (0,1/3)$ or $(s,t)\in (0,1/3)\times(2/3,1)$.
This implies that $\kappa_1$ and $\kappa_2$ are distinct covariance functions that are equal on $\B_{1/3}$.
Moreover, by varying the value of $\lambda\in(0,1)$ we can produce infinitely many distinct $C^{\infty}$ covariances of rank 3, that nevertheless agree on the band of width $1/3$.

\section{Additional Numerical Results} \label{add_res}

Using the simulation setup described in Section \ref{sec:simulations}, we study the efficacy of scree-plot inspection as a means for selecting the rank of $\hat R^K_n$. To this aim, we explore the behaviour of the mapping $i\mapsto f(i)$, for $i=1,\ldots,8$, defined in Section \ref{subsec:scree-plot}, in particular whether the function $f(\cdot)$ indeed levels out at the true rank. In order for functions $f(\cdot)$ corresponding to different $\delta$ values to appear on a same graph, we normalise each function $f(\cdot)$ by $\|P^K_{\delta'} \circ R^K_n \|$, thus rescaling them to a common scale. The results are presented in Fig.~\ref{find_r}. Each plot represents the results for a given scenario and a given rank, and the dotted vertical lines indicate the true value $q$ of the rank. We would select the true rank except in the more challenging cases $\delta\in\{0.5,0.6\}$, where we might select a rank of $2$, though the true rank is  $1$ or $3$. 

\begin{figure}
\centering
\includegraphics[scale=0.7]{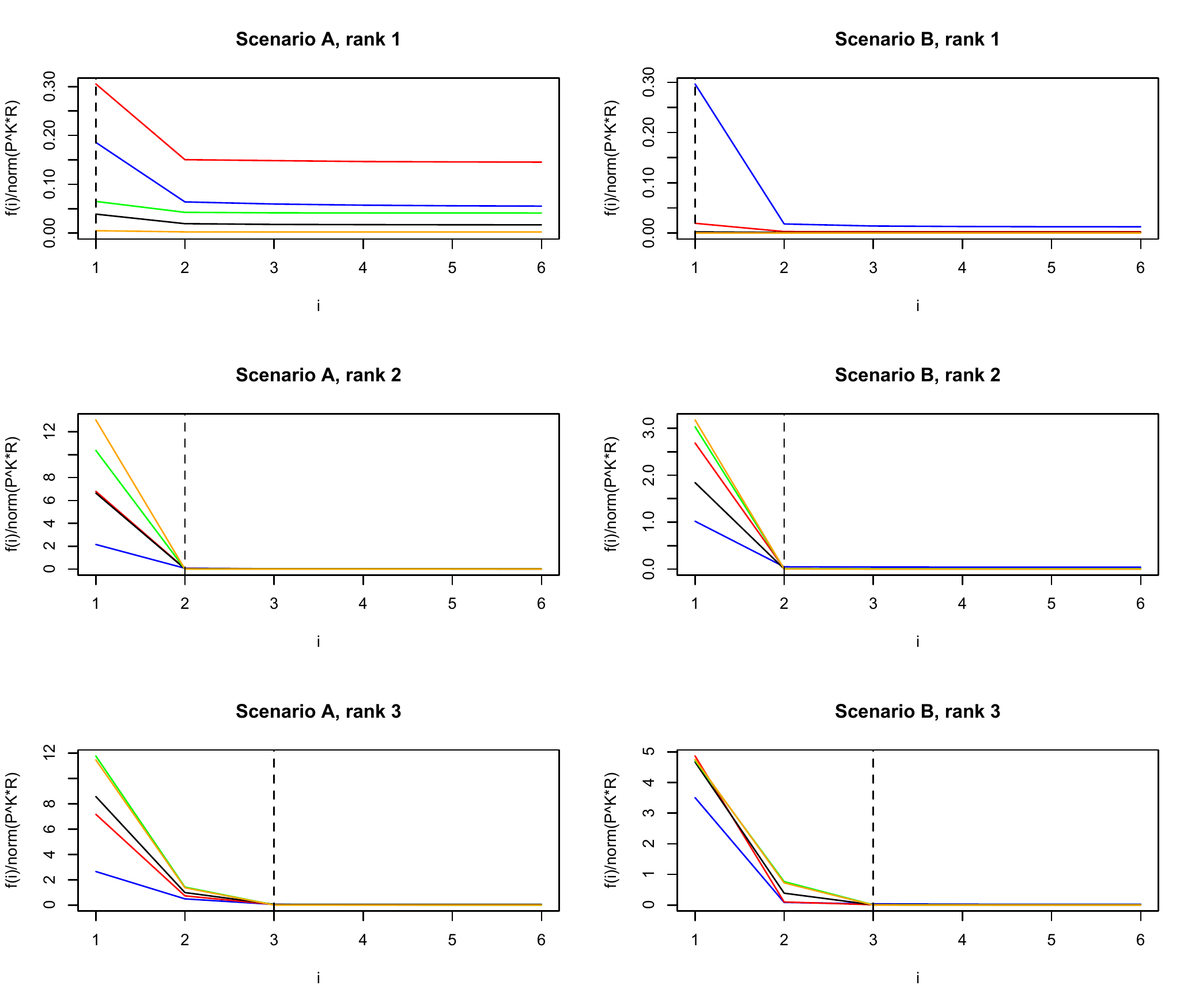}
\caption{Plots of the function $f(\cdot)$, defined in Section \ref{subsec:scree-plot}, normalised by $\|P^K\circ R^K_n \|_F^2$ for a given scenario and a given rank. The curves in blue, red, black, green and orange correspond to settings with $\delta=0.5, 0.6, 0.7, 0.8$ and $0.9$.}
\label{find_r}
\end{figure}

As mentioned in Section \ref{sec:simulations}, we also studied the performance of our method for different values of the grid size $K$. We considered the exact same setup as described in Section \ref{sec:simulations} for Scenario A, but with a grid size equal to $K=25$ or $K=100$. The median, and the first and third quartiles of the $100$ relative error percentages obtained for each combination of parameters are presented in Table \ref{add_resu}. One observes that the results are very similar to those for $K=50$, suggesting a certain amount of stability of the relative error with respect to the grid size.

\begin{table}[ht]
\centering
\begin{tabular}{|c|c|c|c|c|}
\hline
Grid size & $\delta$ $(\delta')$& rank $1$ & rank $2$ & rank $3$ \\
  \hline
  \multirow{5}{*}{$K=25$} 
  &$0.5$ $(0.4)$ & $13$ $(11, 17)$ & $25$ $(18, 31)$ & $35$ $(31, 38)$ \\ 
  &$0.6$ $(0.5)$& $14$ $(10, 17)$ & $16$ $(13, 19)$ & $16$ $(14, 21)$ \\ 
  &$0.7$ $(0.6)$& $12$ $(9, 15)$ & $15$ $(12, 18)$ & $16$ $(13, 20)$ \\ 
  &$0.8$ $(0.7)$& $10$ $(8, 14)$ &$13$ $(10, 15)$ & $14$ $(12, 17)$ \\ 
  &$0.9$ $(0.8)$& $8$ $(6,13)$ & $11$ $(9,16)$ & $12$ $(9,15)$ \\ 
    \hline
   \multirow{5}{*}{$K=100$} 
     &$0.5$ $(0.4)$& $14$ $(12, 19)$ & $28$ $(22, 34)$ & $35$ $(32, 39)$ \\ 
  &$0.6$ $(0.5)$& $15$ $(11, 19)$ & $17$ $(14, 20)$ & $20$ $(16, 23)$ \\ 
  &$0.7$ $(0.6)$& $13$ $(11, 17)$ & $15$ $(13, 18)$ & $17$ $(15, 21)$ \\ 
  &$0.8$ $(0.7)$& $12$ $(10, 17)$ &$15$ $(12, 18)$ & $15$ $(12, 18)$ \\ 
  &$0.9$ $(0.8)$& $10$ $(8,13)$ & $11$ $(9,14)$ & $12$ $(10,15)$ \\  
      \hline
  \end{tabular}
  \caption{Median of the relative error percentage of our estimators for Scenario A for different values of $K$, of the rank and of $\delta$. The first and third quartiles are in parentheses.}
  \label{add_resu}
\end{table}


\bibliographystyle{imsart-nameyear}
\bibliography{biblio}

\end{document}